\documentclass[12pt,draftcls,onecolumn]{IEEEtran}

\newif\ifreport\reporttrue

\IEEEoverridecommandlockouts
\makeatletter
\def\ps@headings{%
\def\@oddhead{\mbox{}\scriptsize\rightmark \hfil \thepage}%
\def\@evenhead{\scriptsize\thepage \hfil \leftmark\mbox{}}%
\def\@oddfoot{}%
\def\@evenfoot{}}
\makeatother
\usepackage[dvips]{epsfig,graphicx}
\usepackage[]{subfigure}
\usepackage{verbatim}
\usepackage{amsmath}
\usepackage{cite}
\usepackage[]{latexsym,times}
\usepackage{verbatim}
\usepackage{mathrsfs}
\usepackage{eucal}
\usepackage{amssymb}
\usepackage{footmisc}
\usepackage{algorithm,algpseudocode}
\usepackage{color}

\ifCLASSINFOpdf
\else
\fi
\long\gdef\boxit#1{\vspace{5mm}\begingroup\vbox{\hrule\hbox{\vrule\kern3pt
\vbox{\kern3pt#1\kern3pt}\kern3pt\vrule}\hrule}\endgroup}

\def\twocolfigure#1#2#3#4#5{
\begin{figure}[htb]
\centerline{\psfig{height=#1pt,width=#2pt,angle=0,figure=./figs/#3.eps}}
\caption{#4} \label{#5}
\end{figure}
}

\newtheorem{deff}{Definition}
\newtheorem{lem}{Lemma}
\newtheorem{thm}{Theorem}

\def\qed{\quad \vrule height5pt width5pt depth0pt}

\def\BibTeX{{\rm B\kern-.05em{\sc i\kern-.025em b}\kern-.08em

    T\kern-.1667em\lower.7ex\hbox{E}\kern-.125emX}}

\begin{document}

\title{Optimal Distributed Resource Allocation for Decode-and-Forward Relay Networks}
\author{Yin Sun, 
        Zhoujia Mao, 
        Xiaofeng Zhong, 
        Yuanzhang Xiao, 
        Shidong~Zhou, 
        and Ness B. Shroff.

\thanks{Y. Sun and Z. Mao are with the Department of Electrical and Computer Engineering, the Ohio State University, 2015 Neil Ave., Columbus, OH 43210, USA. sunyin02@gmail.com, maoz@ece.osu.edu.}
\thanks{X. Zhong and S. Zhou are with the State Key Laboratory
on Microwave and Digital Communications, Tsinghua National
Laboratory for Information Science and Technology, and Department of
Electronic Engineering, Tsinghua University, Beijing, China.
Address: Room 4-407, FIT Building, Tsinghua University, Beijing
100084, People's Republic of China. \{zhongxf,zhousd\}@tsinghua.edu.cn.}
\thanks{Y. Xiao is with the Department of Electrical Engineering, University of California, Los Angeles. e-mail: xyz.xiao@gmail.com.}
\thanks{N. B. Shroff is with the Department of Electrical and Computer Engineering and the Department of Computer Science and Engineering, the Ohio State
University, 2015 Neil Ave., Columbus, OH 43210, USA. shroff@ece.osu.edu.}
\thanks{The material in this paper was presented in part in IEEE GLOBECOM 2009.}}
\maketitle
\vspace{-1.5cm}
\begin{abstract}
This paper presents a distributed resource allocation algorithm
to jointly optimize the power allocation, channel allocation and relay selection for decode-and-forward (DF) relay networks with a large number of sources, relays, and destinations. The well-known dual decomposition technique cannot directly be applied to resolve this problem, because the achievable data rate of DF relaying is not strictly concave, and thus the local resource allocation subproblem may have non-unique solutions. We resolve this non-strict concavity problem by using the idea of the proximal point method, which adds quadratic terms to make the objective function strictly concave. However, the proximal solution adds an extra layer of iterations over typical duality based approaches, which can significantly slow down the speed of convergence. To address this key weakness, we devise a fast algorithm without the need for this additional layer of iterations, which converges to the optimal solution. Our algorithm only needs local information exchange, and can easily adapt to variations of network size and topology. We prove that our distributed resource allocation algorithm converges to the optimal solution. A channel resource adjustment method is further developed to provide more channel resources to the bottleneck links and realize traffic load balance. Numerical results are provided to illustrate the benefits of our algorithm.
\end{abstract}


\section{Introduction}
\label{sec:intro}

Cooperative relaying has recently received a lot of attention as a promising technique to improve the throughput, coverage, and reliability of wireless networks \cite{Sendonaris2003,Laneman04}. The decode-and-forward (DF) relay strategy has been advocated by several standard organizations for next generation wireless networks \cite{Peters09,Loa2010,OFDM_relay_overview10}. In this strategy, a source node, a relay node, and a destination node cooperate to form a DF relay link. The relay node decodes the source node's transmission message, then forwards the recovered message to the destination.


Since power and channel (code, time, and/or frequency) are crucial resources in wireless networks, a number of studies have investigated the resource allocation for DF relay networks. They have shown that optimal resource allocation can achieve significant performance improvement for DF relay networks with a single source-destination data stream \cite{Madsen05,Zhao07,Liang07,Hong07,XiZhang_Jsac07,Yonghui07,
Ding08,Chen108,Dang10,Maric10,YinSunTSP11}. However, optimal resource allocation becomes much more challenging in scenarios with many sources, relays, and destinations, because each source-destination data stream may cooperate with several relay nodes and each relay node can assist several data streams, which could form a large number of potential DF relay links. In order to achieve higher network throughput, appropriate DF relay links should be selected out and allocated with suitable source power, relay power, and channel resources for transmissions.
%
%



\begin{figure}[!t]
    \centering
        \resizebox{0.35\textwidth}{!}{\includegraphics{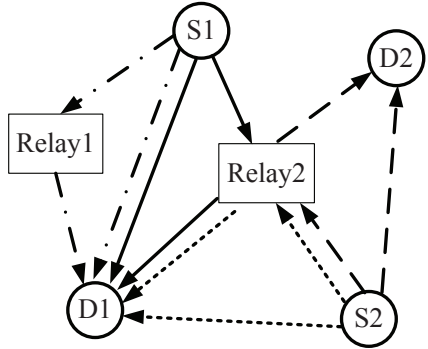}}
        \vspace{-0.0cm}
        \caption{A DF relay network with 3 source-destination data streams and 2 relays.}
        \label{fig1}
\end{figure}

Centralized power allocation and relay selection algorithms of relay networks were proposed in \cite{Guoqing06,WeiYuJsac07,Lingfan09,Jitvanichphaibool09}, which requests signaling mechanisms to gather the channel state information (CSI) of all the wireless links at a central control node. However, such mechanisms are difficult to implement in practice and will especially not work well when the network size is large.
Recently, a great deal of research efforts have focused on distributed resource allocation for DF relay networks: Traditional two-hop relaying was considered in \cite{YingCui09,Hongxing09,Gong11}, which ignores the source-destination wireless channel, and thus achieves a lower data rate. In \cite{Long08,Kadloor10,Hosseini10}, several other conservative rate functions of DF relaying were employed to simplify the distributed resource allocation problem and make it tractable. Optimal power allocation of the relay nodes was studied in \cite{Guoqing06,WeiYuJsac07,Serbetli08,Lingfan09,Jitvanichphaibool09} without considering power allocation of the source node and relay selection. Optimal distributed resource allocation for DF relay networks remains to be a difficult and crucial problem.

Dual decomposition techniques have been effectively used in multi-hop wireless networks to achieve optimal resource allocation results (e.g., see \cite{LinJsac06} and the references therein). 
However, such techniques cannot be applied directly to DF relay networks ---- the local resource allocation subproblems may have non-unique optimal solutions, because the achievable rate of DF relaying is not strictly concave. Since no global network information is available to the local resource allocation solvers, it is quite difficult to find a global feasible solution among all the locally optimal solutions \cite{Xiao04,Yu06}.

One promising method to address this non-strict concavity problem is the {proximal point method} \cite{Bertsekas89}, which adds strictly concave terms to the achievable rate function without affecting the optimal solution. However, typical proximal point algorithms require an extra outer layer of iterations compared to the conventional dual decomposition based algorithms \cite{Bertsekas89}, which, in turn, results in a slower convergence speed.

In this paper, we investigate the optimal resource allocation of DF relay networks, which may have a large number of sources, relays, and destinations. Each source node may transmit to one or several destinations through the assistance of several relay nodes. Meanwhile, each relay node may aid several source-destination data streams, as illustrated in Fig. \ref{fig1}. Each of the source and relay nodes has an individual transmission power constraint. The \emph{channel resources (code, time, and/or frequency)} of the network are managed by many distributed control nodes. Each control node is pre-assigned with some channel resources, and is responsible to allocate these channel resources to nearby wireless links. The main contributions of this paper are summarized as follows:

\begin{itemize}
\item We propose a distributed resource allocation algorithm that jointly optimizes the power allocation, channel allocation and relay selection of DF relay networks so as to maximize the total throughput. The candidate DF relay links compete for the channel resources and the DF relay links with zero (or very few) channel resources are not selected for transmission. By this, the optimal relay selection is obtained through channel allocation. Our algorithm requires less layers of iterations than traditional proximal point algorithms, and thus have a much faster convergence speed. In addition, our algorithm only needs local information exchange among the source, relay, destination, and control nodes of each DF relay link, and can easily adapt to variations of network size and topology. We prove that this algorithm converges to the optimal resource allocation solution.

\item In practice, the spatial distribution of wireless traffic is
usually non-uniform and varies from time to time. The
pre-assigned channel resources of the distributed control
nodes may be inadequate to support the heavy wireless
traffic. To address this problem, we develop a centralized
channel resource adjustment algorithm on top of
the proposed distributive resource allocation algorithm.
The bottleneck control nodes are provided with more
channel resources in order to balance the traffic load.
\end{itemize}

The proposed power allocation algorithm is motivated by the work in \cite{Lin06}, where a modified proximal point algorithm with less iteration layers was proposed for multi-path routing problems. However, the structure of the objective function in that work is very different from our work that deals with DF relay networks. Hence, a substantially new proof methodology is required to show convergence in our context, which is one of the major contributions of this paper.

The remaining parts of this paper are organized as follows: In
Section \ref{sec:model}, we present the system model and the
formulation of power allocation problem. In Section \ref{sec3}, we describe our distributed power allocation algorithm. A centralized channel resource adjustment algorithm is presented in Section \ref{sec4}. Some simple extensions of our algorithm are discussed in Section \ref{sec6}. Numerical results are provided in Section \ref{sec5}, and we conclude the paper in Section \ref{conclusion}.

\section{System Model and Problem Formulation}

\label{sec:model}
Consider a DF relay network with $N$ source/destination nodes, denoted by the set $\mathcal{N}=\{1,2,\ldots,N\}$, and $J$ relay nodes, represented by the set $\mathcal{J}=\{1,2,\ldots,J\}$. Each source-destination data stream in the network is denoted as $m=(s,d)$ with $s,d\in\mathcal{N}$. The set of all data streams is denoted by $\mathcal{M}\subseteq \{(i,j)| i,j\in \mathcal{N}, i\neq j\}$.
The $m$th data stream either can transmit directly through the source-destination wireless channel, or can be assisted by $J(m)$ candidate relay nodes and form $J(m)$ possible DF relay links. The set of candidate relay nodes for the $m$th data stream is denoted by $\mathcal{J}(m) \subseteq \mathcal{J}$.
We assume that each DF relay link only involves one relay node. When several relay nodes assist the same source-destination data stream, they belong to different DF relay links to avoid the implementation complexity due to the cooperation among different relay nodes. We further assume that the direct transmission (DT) link and DF relay links of each data stream operate over orthogonal wireless channels by means of code division multiple access (CDMA), or frequency/time division multiple access (F/TDMA) as in \cite{Serbetli08,Lingfan09}.


The wireless transmissions in this network are managed by $T$ distributed control nodes, which are denoted by $\mathcal{T}=\{1,2,\ldots,T\}$. The candidate wireless links of data stream $m$, including $J(m)$ DF relay links and one direct transmission (DT) link, are managed by the control node $c(m)$. We assume that neighboring control nodes are assigned to orthogonal channels to suppress co-channel interference, and distant control nodes are allowed to reuse the same channel resources. Such an assumption is practical for many wireless networks. For example, in wireless sensor networks, the traffic load and the transmission power are very small and hence interference is less of an issue \cite{cooperative_sensor10,Kim2010}.

The DF relay procedure consists of two phases: In Phase 1, the source node transmits a message to the relay and destination nodes. The relay node decodes its received message, while the destination stores its received signal for later decoding. In Phase 2, the relay node forwards the recovered message to the destination. The destination combines its received signals in two phases to decode the source node's message \cite{Laneman04}.
Let $h_m^{s,d}$ denotes the complex channel coefficient of the source-destination wireless link of data stream $m$, $h_{mj}^{s,r}$ and $h_{mj}^{r,d}$ denote the complex channel coefficients of the source-relay and relay-destination links of the DF relay link composed by data stream $m$ and relay node $j$.

The spectrum efficiency of the DT link of data stream $m$ is given by the capacity of Gaussian channel, i.e.,
\begin{eqnarray}
\label{eq38}
R_m^{DT}\!=\!\theta_{m}^{DT}\!\log_2\!\left(\!1\!+\!\frac{P^s_{m}|h^{s,d}_{m}|^2}{\theta_{m}^{DT}N_0W}\!\right)
\!=\!\theta_{m}^{DT}\!\log_2\!\left(\!1\!+\!\frac{P^s_{m}g^{s,d}_{m}}{\theta_{m}^{DT}}\!\right)\!\!,\!\!\!
\end{eqnarray}
where $P^s_{m}\geq0$ is transmission power of the source node, $\theta_m^{DT}$ is the corresponding proportion of channel resources, $W$ is the total amount of available channel resources, $N_0$ is power spectral density of the Gaussian noise at each receiver,
and $g_m^{s,d}\triangleq \frac{|h_m^{s,d}|^2}{N_0 W}$ characterizes the quality of the source-destination wireless channel of the $m$th data stream, as shown in Fig. \ref{1fig2}. The spectrum efficiency achieved by the DF relay link composed by data stream $m$ and relay node $j$ can be described as \cite{Laneman04}:
\begin{eqnarray}\label{eq1}
R_{mj}^{DF}\!\!\!\!\!\!\!\!\!\!
&&=\frac{\theta_{mj}^{DF}}{2}\min\left\{
\log_2\left(1+\frac{2P^s_{mj}|h^{s,r}_{mj}|^2}{\theta_{mj}^{DF}N_0W}\right),\right.\nonumber\\
&&~~~~~~~~~\left.
\log_2\left[1+\frac{2\big(P^s_{mj}|h^{s,d}_{m}|^2+P^r_{mj}|h^{r,d}_{mj}|^2\big)}{\theta_{mj}^{DF}N_0W}\right]
\right\}\nonumber\\
&&=\frac{\theta_{mj}^{DF}}{2}\min\left\{
\log_2\left(1+\frac{2P^s_{mj}g^{s,r}_{mj}}{\theta_{mj}^{DF}}\right),\right.\nonumber\\
&&~~~~~~~~~\left.\log_2\left[1+\frac{2\big(P^s_{mj}g^{s,d}_{m}+P^r_{mj}g^{r,d}_{mj}\big)}{\theta_{mj}^{DF}}\right]
\right\},
\end{eqnarray}
where $P^s_{mj},P^r_{mj}\geq0$ are the transmission powers of the source and relay nodes, $\theta_{mj}^{DF}$ is the corresponding proportion of channel resources, $g_{mj}^{s,r}\triangleq \frac{|h_{mj}^{s,r}|^2}{N_0W}$ and
$g_{mj}^{r,d}\triangleq \frac{|h_{mj}^{r,d}|^2}{N_0W}$ characterize the quality of the source-relay and relay-destination wireless channels of this DF relay link.

\begin{figure}[!t]
    \centering
        \resizebox{0.35\textwidth}{!}{\includegraphics{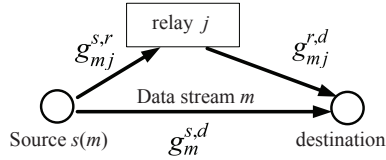}}
        \vspace{-0.3cm}
        \caption{Illustration of the DF relay strategy.}
        \label{1fig2}
        \vspace{-0.5cm}
\end{figure}

Let $s(m)$ represent the source node of data stream $m$. Then, the power constraint of source node $l$ over all the channels can be determined as
\begin{align*}
\sum_{\{m|s(m)=l\}}\left(P^s_{m}+\sum_{j\in\mathcal
J(m)} P^s_{mj}\right)\leq P^s_{l,\max},
\end{align*}
\noindent where $P^s_{l,\max}$ is the maximal transmission power of source node $l$. The power constraint of relay node $j$ is expressed as
\begin{align*}
\sum_{\{m|j\in\mathcal J(m)\}}P^r_{mj} \leq P^r_{j,\max},
\end{align*}
\noindent where $P^r_{j,\max}$ is the maximal transmission power of relay node $j$.
The local channel resource constraint managed by control node $t$
is given by
\begin{align*}
\sum_{\{m|c(m)=t\}}\left(\theta_{m}^{DT}+\sum_{j\in\mathcal
J(m)}\theta_{mj}^{DF}\right) \leq \beta_t,
\end{align*}
\noindent where $\beta_t$ represent the proportion of
channel resources pre-assigned to control node $t$.

In practice, the channel codes with short block length usually have quite poor error performance \cite{BK:Tse05}. The code block length and corresponding occupied channel resources should be large enough to guarantee a small decoding error probability, and the links with very few channel resources are not admitted for transmission. This requirement can be expressed by the following channel resource constraint:
\begin{align}\label{eq2}
\theta_{m}^{DT},\theta_{mj}^{DF}\geq\theta_{\min},
\end{align}
where $\theta_{\min}>0$ is a very small constant.
If the channel resource fraction of a DF relay link or a DT link is equal to $\theta_{\min}$, the network will not admit the transmission of this link.
In such a way, the optimal relay selection is fulfilled. From the perspective of optimization, setting the minimum value $\theta_{\min}$ can avoid the channel resources from approaching zero, and thus preventing unnecessary technical complications due to non-differentiable achievable rates $R_m^{DT}$ and $R_{mj}^{DF}$.


The joint design of power allocation, channel allocation, and relay selection of the DF relay network is formulated as
\begin{subequations}
\begin{align}
\!\!\!\!(\textrm{P}) \max
\limits_{\substack{P^s_{m},P^s_{mj},P^r_{mj},\theta_{m}^{DT},\theta_{mj}^{DF}}}&\sum_{m=1
}^M\left(R_m^{DT}+\sum_{j\in\mathcal J(m)} R_{mj}^{DF}\right)\label{eq5}\\
\rm{s.t.}~~~~~&\!\!\!\!\!\!\!\!\sum_{\{m|s(m)=l\}}\left(P^s_{m}+\sum_{j\in\mathcal
J(m)} \!\!P^s_{mj}\right)\!\leq\!
P^s_{l,\max}, \forall~l \label{eq10}\!\!\\
&\!\!\!\!\!\!\!\!\sum_{\{m|j\in\mathcal J(m)\}}P^r_{mj} \leq
P^r_{j,\max}, \forall~j \label{eq11}\\
&\!\!\!\!\!\!\!\!\sum_{\{m|c(m)=t\}}\left(\theta_{m}^{DT}+\sum_{j\in\mathcal J(m)}\theta_{mj}^{DF}\right) \leq \beta_t,~ \forall~t \label{eq31}\\
&P^s_{m},P^s_{mj},P^r_{mj}\geq0,\forall~m,j,\\
&\theta_{m}^{DT},\theta_{mj}^{DF}\geq\theta_{\min},\forall~m,j.\label{}
\end{align}
\end{subequations}
Note that although all the candidate DT and DF relay links are included in the objective function of Problem $(P)$, only the wireless links with channel resources larger than $\theta_{\min}$ are admitted for transmission after solving Problem $(P)$.
If $g_{mj}^{s,r}\leq g_{m}^{s,d}$, one can simply show that $R_{mj}^{DF}<R_{m}^{DT}$. Hence, DF relaying cannot achieve a higher data rate than DT transmission in this case. Therefore, only the relay nodes satisfying $g_{mj}^{s,r}>g_{m}^{s,d}$ need to be considered in the candidate relay set $\mathcal{J}(m)$ for date stream $m$.

%
The achievable rates $R_m^{DT}$ and $R_{mj}^{DF}$ are both concave in their power and channel resource variables. Therefore, the resource allocation problem $(\textrm{P})$ is a convex optimization problem.
However, as we have mentioned earlier, the achievable rate of DF relaying $R_{mj}^{DF}$ is not strictly concave. Specifically, $R_{mj}^{DF}$ is linear in the transmission power variables $P^s_{mj}$ and $P^r_{mj}$ in two cases: When  $P^s_{mj}g^{s,r}_{mj}<P^s_{mj}g^{s,d}_{m}+P^r_{mj}g^{r,d}_{mj}$ holds, the achievable rate
$R_{mj}^{DF}$ in \eqref{eq1} does not vary with respect to $P^r_{mj}$. Moreover, if $P^s_{mj}g^{s,r}_{mj}>P^s_{mj}g^{s,d}_{m}+P^r_{mj}g^{r,d}_{mj}$ holds and the value of $P^s_{mj}g^{s,d}_{m}+P^r_{mj}g^{r,d}_{mj}$ is fixed, $R_{mj}^{DF}$ maintains the same value as $P^r_{mj}$ varies.

In dual decomposition based distributed optimization techniques, it is quite difficult to recover the optimal primal variables (i.e., the transmission power variables $P^s_{mj}$ and $P^r_{mj}$), if the objective function is non-strictly concave \cite{BK:Bertsekas,Xiao04,Yu06,Long08}. Although the dual variables converge to the optimal solution to the dual problem, the primal variables may oscillate forever and never result in a feasible solution \cite{Lin06}.
In the next section, we develop a distributed power allocation algorithm to address this non-strict concavity difficulty, and then prove its convergence to the optimal solution.

\section{Distributed Resource Allocation}
\label{sec3}
To circumvent this non-strict concavity difficulty, we use the idea of proximal point method \cite{Bertsekas89}, which is to add some quadratic terms and make the objective function strictly concave in the primal variables. However, a standard proximal point method is not effective because it relies on an extra outer layer of iterations, and hence results in a slow convergence speed. We will overcome this difficulty by developing a distributed resource allocation algorithm without increasing the number of iteration layers. The details are provided below.

\subsection{Distributed Resource Allocation Algorithm}
The original problem~$(\textrm{P})$ is rewritten as the following problem with some extra auxiliary variables:
\begin{subequations}
\label{eq7}
\begin{align}
&\!\!\!\!\!\!\!\!\!\!\!\!\!\!\!\!\!\!\!\!\!\!\!\!\!\!\!\max
\limits_{\substack{P^s_{m},P^s_{mj},P^r_{mj},\theta_{m}^{DT},\theta_{mj}^{DF}\\Q^{s}_{m},Q^{s}_{mj},Q^{r}_{mj}}}
~~\sum_{m=1
}^M\left[R_m^{DT}-\frac{c_{m}}{2}\left(P^s_{m}-Q^{s}_{m}\right)^2\right]\label{eq79}\nonumber\\
&~+\sum_{m=1
}^M\sum_{j\in\mathcal J(m)}\left[ R_{mj}^{DF}-\frac{c_{mj}}{2}\left(P^s_{mj}-Q^{s}_{mj}\right)^2\right.\nonumber\\
&~~~~~~~~~~~~~~~~~~~~~~~~~\left.-\frac{c_{mj}}{2}\left(P^r_{mj}-Q^{r}_{mj}\right)^2\right]\\
~~~~~~~\rm{s.t.}&\sum_{\{m|s(m)=l\}}\left(P^s_{m}+\sum_{j\in\mathcal
J(m)} P^s_{mj}\right)\!\leq\!
P^s_{l,\max}, \forall~l\!\! \label{eq33}\\
&\sum_{\{m|j\in\mathcal J(m)\}}P^r_{mj} \leq
P^r_{j,\max}, \forall~j \label{eq34}\\
&\sum_{\{m|c(m)=t\}}\left(\theta_{m}^{DT}+\sum_{j\in\mathcal J(m)}\theta_{mj}^{DF}\right) \leq \beta_t,~ \forall~t \label{eq15}\\
&~~~~P^s_{m},P^s_{mj},P^r_{mj}\geq0,\forall~m,j,\\
&~~~~\theta_{m}^{DT},\theta_{mj}^{DF}\geq\theta_{\min},\forall~m,j,\label{eq16}
\end{align}
\end{subequations}
where $c_{m},c_{mj}>0$ are algorithm parameters, $Q^{s}_{m}$, $Q^{s}_{mj}$, and $Q^{r}_{mj}$ are auxiliary variables corresponding to $P^s_{m}$, $P^s_{mj}$, and $P^r_{mj}$, respectively. It is easy to show that the optimal value of \eqref{eq79} coincides with that of \eqref{eq5} \cite{Bertsekas89}. In fact, let $\vec{P}^\star$ denote the maximizer of $(\textrm{P})$, then $\vec{P}=\vec{P}^\star$, $\vec{Q}=\vec{P}^\star$ maximizes \eqref{eq7}.
Moreover, problem \eqref{eq7} is strictly concave with respect to the transmission power variables $P^s_{m}$, $P^s_{mj}$, and $P^r_{mj}$. 
In the sequent, we solve problem \eqref{eq7} instead of the original problem $(\textrm{P})$.


Let $\mu_l$ and $\nu_{j}$ be the Lagrange multipliers
associated with the constraints in \eqref{eq33} and \eqref{eq34}, respectively. The partial Lagrangian of problem~\eqref{eq7} with respect to the power constraints \eqref{eq33} and \eqref{eq34} is given by
\begin{align}
\label{eq12}
&L\left(P^s_{m},P^s_{mj},P^r_{mj},Q^{s}_{m},Q^{s}_{mj},Q^{r}_{mj},\theta_{m}^{DT},\theta_{mj}^{DF};\mu_{l},\nu_{j}\right)\nonumber\\
=&\sum_{m=1
}^M\bigg\{\left[R_m^{DT}-\frac{c_{m}}{2}\left(P^s_{m}-Q^{s}_{m}\right)^2\right]\nonumber\\
&+\sum_{j\in\mathcal J(m)}\left[ R_{mj}^{DF}-\frac{c_{mj}}{2}\left(P^s_{mj}-Q^{s}_{mj}\right)^2\right.\nonumber\\
&\left.~~~~~~~~~~~~~~~~~~-\frac{c_{mj}}{2}\left(P^r_{mj}-Q^{r}_{mj}\right)^2\right]\bigg\}\nonumber\\
&-\sum_{l=1}^N
\mu_{l}\left[\sum_{\{m|s(m)=l\}}\left(P^s_{m}+\sum_{j\in\mathcal J(m)}P^s_{mj}-
P^s_{l,\max}\right)\right]\nonumber\\
&-\sum_{j=1}^J\nu_{j}\left(\sum_{\{m|j\in\mathcal J(m)\}} P^r_{mj} -
P^r_{j,\max}\right).
\end{align}

For convenience, we rearrange the above Lagrangian as
\begin{align}\label{eq50}
&L\left(\vec{P},\vec{Q},\vec{\theta};\vec{\nu}\right)\nonumber\\
=&R\left(\vec{P},\vec{\theta}\right)\!-\!\frac{1}{2}\left(\vec{P}\!-\!\vec{Q}\right)^TV
\left(\vec{P}\!-\!\vec{Q}\right)\!-\!\vec{\nu}^T\left(E\vec{P}\!-\!\vec{P}_{\max}\right),\!\!\!
\end{align}
\noindent where $R(\vec{P},\vec{\theta})$ is the objective function of problem $(\textrm{P})$, $\vec{\theta}$ is a
$M+\sum_{m=1}^MJ(m)$ dimensional vector representing the channel resources of the DT and DF relay links
$\theta_{m}^{DT}$ and $\theta_{mj}^{DF}$,
$\vec{P}$ is a $M+2\sum_{m=1}^MJ(m)$ dimensional vector representing the power allocation variables
$P^{s}_{m}$, $P^{s}_{mj}$, and $P^{r}_{mj}$,
$\vec{Q}$ is a $M+2\sum_{m=1}^MJ(m)$ dimensional vector representing the auxiliary variables
$Q^{s}_{m}$, $Q^{s}_{mj}$, and $Q^{r}_{mj}$, $\vec{P}_{\max}$ is a $N+J$ dimensional vector representing the maximal transmission power $P^s_{l,\max}$ and $P^r_{l,\max}$, $\vec{\nu}$ is a $N+J$ dimensional vector representing the
dual variables $\mu_{l}$ and $\nu_{j}$, and $E$ is a $(N+J)\times(M+2\sum_{m=1}^MJ(m))$ matrix representing the relationship between the transmitting power variables and corresponding source/relay nodes.

We now present our resource allocation algorithm and then describe its distributed implementation procedure in the next subsection.

\textbf{Algorithm}~$\mathcal{A}$:
\begin{enumerate}
\item The source and relay nodes select the power allocation policy for the $k$th iteration by
\begin{eqnarray}\label{eq14}
\vec{x}(k)=\arg\max
\limits_{\vec{P}\geq0}L\left(\vec{P},\vec{Q}(k),\vec{\theta};\vec{\nu}(k)\right).
\end{eqnarray}
The source and relay nodes update the dual variables $\vec{\nu}(k+1)$ according to the following equation:
\begin{eqnarray}\label{eq20}
\vec{\nu}(k+1)=\left\{\vec{\nu}(k)+A\left[E\vec{x}(k)-\vec{P}_{\max}\right]\right\}^+,
\end{eqnarray}
\noindent where $A$ is a $(N+J)\times(N+J)$ dimensional diagonal matrix with diagonal elements $\alpha_l$~$(l=1,2,\cdots, N+J)$ as the step-size of dual updates, and $(\cdot)^+ \triangleq \max\{\cdot,0\}$.

\item
The source and relay nodes update the auxiliary variable $\vec{Q}(k+1)$ by
\begin{eqnarray}
\label{eq23}
\vec{Q}(k\!+\!1)=\arg\max
\limits_{\vec{P}\geq0}L\left(\vec{P},\vec{Q}(k),\vec{\theta};\vec{\nu}(k\!+\!1)\right).
\end{eqnarray}

\item Periodically, after every $K$ iterations, the control node set the channel allocation by
\begin{eqnarray} \label{eq13}
\vec{\theta}\leftarrow\arg\max_{\vec{\theta}\in\mathcal {W}}
R\left(\vec{Q}(k+1),\vec{\theta}\right),
\end{eqnarray}
where $\mathcal{W}$ denotes the feasible set of channel resources described by \eqref{eq15} and \eqref{eq16}.
\end{enumerate}

In the traditional proximal point method \cite{Bertsekas89}, Step 2) is implemented only after many executions of Step 1), i.e. until the updates of \eqref{eq14} and \eqref{eq20} have converged, which results in a three-layer iterative algorithm. On the other hand, our Algorithm $\mathcal{A}$ only requires a two-layer iteration structure, and hence has a faster convergence speed.

%
%

\subsection{Distributed Implementation of Algorithm $\mathcal{A}$}
We proceed to show that each step of Algorithm $\mathcal{A}$ can be fulfilled in a distributed fashion, which only requires local information exchange among the source, relay, destination, and control nodes of each DF relay link.

First, the dual update \eqref{eq20} can be equivalently expressed as
\begin{eqnarray}
&&\!\!\!\!\!\!\!\!\!\!\!\mu_l(k\!+\!1) \!=\! \left[\mu_l(k)\!+\!a_l \left(P^s_{m}\!+\!\sum_{j\in\mathcal J(m)}P^s_{mj}\!-\!
P^s_{l,\max}\right)\right]^+\!\!,~~~~\\
&&\!\!\!\!\!\!\!\!\!\!\!\nu_j(k\!+\!1) \!=\! \left[\nu_j(k)\!+\!a_{N+j} \left(\sum_{j\in\mathcal J(m)}P^r_{mj}\!-\!
P^r_{j,\max}\right)\right]^+\!\!,
\end{eqnarray}
which can be carried out distributedly at each source and relay node.

In addition, the Lagrangian maximization problems in \eqref{eq14} and \eqref{eq23} can be decomposed into many independent local power allocation subproblems. Specifically, the terms of the Lagrangian $L$ in \eqref{eq12} can be reassembled as
\begin{align}
&L\left(\vec{P},\vec{Q},\vec{\theta};\vec{\nu}\right)\nonumber\\
=&\sum_{m=1}^M\bigg\{\left[R_m^{DT}-\frac{c_{m}}{2}\left(P^s_{m}-Q^{s}_{m}\right)^2-\mu_{s(m)}P^s_{m}\right]\nonumber\\
&+\sum_{j\in\mathcal J(m)}\!\!\left[R_{mj}^{DF}\!-\!\frac{c_{mj}}{2}\left(P^s_{mj}\!-\!Q^{s}_{mj}\right)^2\!-\!\frac{c_{mj}}{2}\left(P^r_{mj}\!-\!Q^{r}_{mj}\right)^2\right.\nonumber\\
&\left.-\mu_{s(m)}P^s_{mj}-\nu_{j}P^r_{mj}\right]\bigg\}+\sum_{l=1}^N\mu_{l}P^s_{l,\max}+\sum_{j=1}^J\nu_{j}P^r_{j,\max}.
\end{align}
Therefore, the Lagrangian maximization problem in \eqref{eq14} and \eqref{eq23} can be rewritten as
\begin{align}
&\max_{\vec{P}\geq0} L\left(\vec{P},\vec{Q},\vec{\theta};\vec{\nu}\right)  \nonumber\\ =&\sum_{m=1}^M\limits\!\!\left[ H_m\left(\theta_m^{DT}, Q^{s}_{m};\mu_{s(m)}\right)\!\!+\!\! \sum_{j\in\mathcal J(m)}\!\! I_{mj}\left(\theta_m^{DF},Q^{s}_{mj},Q^{r}_{mj};\mu_{s(m)},\nu_{j}\right)\right]\nonumber\\
&+\sum_{l=1}^N\mu_{l}P^s_{l,\max}+\sum_{j=1}^J\nu_{j}P^r_{j,\max},
\end{align}
where
\begin{align}
\label{eq17}
&H_m\left(\theta_m^{DT}, Q^{s}_{m};\mu_{s(m)}\right) \nonumber\\
= & \max_{\substack{P^s_{m}\geq0}}R_m^{DT}-\frac{c_{m}}{2}\left(P^s_{m}-Q^{s}_{m}\right)^2-\mu_{s(m)}P^s_{m},\\
\label{eq18}
&I_{mj}\left(\theta_m^{DF},Q^{s}_{mj},Q^{r}_{mj};\mu_{s(m)},\nu_{j}\right)\nonumber\\
=& \max_{\substack{P^s_{mj},P^r_{mj}\geq0}}R_{mj}^{DF}\!-\!\frac{c_{mj}}{2}\left(P^s_{mj}\!-\!Q^{s}_{mj}\right)^2
\nonumber\\
&-\!
\frac{c_{mj}}{2}\left(P^r_{mj}\!-\!Q^{r}_{mj}\right)^2\!-\!\mu_{s(m)}P^s_{mj}\!-\!\nu_{j}P^r_{mj},
\end{align}
are local power allocation subproblems for the DT link and DF relay link, respectively.
The closed-form solutions to \eqref{eq17} and \eqref{eq18} are provided in the following lemmas, where the subscripts are omitted for ease of notation:
\begin{lem}\label{lem4}
The optimal solution to \eqref{eq17} is
\begin{align}
\label{eq52}
P^{s}=f(2\theta^{DT},c,\mu,Q^s,g^{s,d},1),
\end{align}
\noindent where
\begin{align}\label{eq74}
&f(\theta,c,\mu,Q,g,v)\!\triangleq\!\frac{1}{2}\left(\frac{\theta}{\mu\ln2}\!-\!\frac{\theta}{g}\!+\!\sqrt{x^2\!+\!y}\!-\!x\right)^+,\\
&x =
\frac{\mu}{cv}-\frac{Q}{v}+\frac{\theta}{\mu\ln2}-\frac{\theta}{2g},\\
&y =\frac{2Q\theta}{\mu
v\ln2}+ \frac{\theta^2}{g\mu\ln2}-\frac{\theta^2}{\mu^2\ln2^2}.
\end{align}
\end{lem}
The proof of Lemma \ref{lem4} is provided in Appendix~\ref{distributedpowerallocation}. Note that $\sqrt{x^2+y}-x$ tends to 0 as $c\rightarrow0$. At this limit, \eqref{eq52} reduces to the conventional water-filling solution.

\begin{lem}\label{lem5}
The optimal solution to \eqref{eq18} is provided for three separate cases:

\noindent\textbf{Case 1:} if $g^{r,d}P^r>(g^{s,r}-g^{s,d})P^s$, the optimal values of $P^s$ and $P^r$ are given by
\begin{equation}\label{eq53}
\left\{
\begin{array}{l}
P^{s}=f(\theta^{DF},c,\mu,Q^s,g^{s,r},1),\\
P^{r}=\left[-\nu_{}/c+Q^r\right]^+.
\end{array}\right.
\end{equation}
\textbf{Case 2:} $g^{r,d}P^r<(g^{s,r}-g^{s,d})P^s$. If $P^r$ in \eqref{eq63} satisfies $P^r\geq0$, the optimal values of $P^s$ and $P^r$ are given by
\begin{equation}\label{eq63}
\left\{
\begin{array}{l}
P^s = \frac{g^{r,d}(g^{r,d}\nu-g^{s,d}\mu)}{\left[\left(g^{s,d}\right)^2+\left(g^{r,d}\right)^2\right]c} +\frac{g^{r,d}(g^{s,d}Q^s-g^{r,d}Q^r)}{\left(g^{s,d}\right)^2+\left(g^{r,d}\right)^2}\\
~~~~~~~~~~~~~~~~~~~~~~~~~~~~~+\frac{g^{s,d}e}{\left(g^{s,d}\right)^2+\left(g^{r,d}\right)^2},\\
P^r = -\frac{g^{s,d}(g^{r,d}\nu-g^{s,d}\mu)}{\left[\left(g^{s,d}\right)^2+\left(g^{r,d}\right)^2\right]c}
-\frac{g^{s,d}(g^{s,d}Q^s-g^{r,d}Q^r)}{\left(g^{s,d}\right)^2+\left(g^{r,d}\right)^2}\\
~~~~~~~~~~~~~~~~~~~~~~~~~~~~~+\frac{g^{r,d}e}{\left(g^{s,d}\right)^2+\left(g^{r,d}\right)^2},
\end{array}\right.
\end{equation}
where $e$ is the value of $P^s_{}g^{s,d}_{}+P^r_{}g^{r,d}_{}$ given by
\begin{align}\label{eq73}
e=&f\left(\theta^{DF}\left[\left(g^{s,d}\right)^2+\left(g^{r,d}\right)^2\right],
c,g^{s,d}\mu+g^{r,d}\nu,\right.\nonumber\\
&~~~\left.g^{s,d}{Q}^s+g^{r,d}{Q}^r,\left(g^{s,d}\right)^2+\left(g^{r,d}\right)^2,1\right),
\end{align}
Otherwise, if $P^r$ in \eqref{eq63} is negative, the optimal values of $P^s$ and $P^r$ are given by
\begin{equation}\label{eq54}
\left\{
\begin{array}{l}
P^{s}=f(\theta^{DF},c,\mu,Q^s,g^{s,d},1),\\
P^{r}=0.
\end{array}\right.
\end{equation}
\textbf{Case 3:} if $g^{r,d}P^r=(g^{s,r}-g^{s,d})P^s$, the optimal values of $P^s$ and $P^r$ are given by
\begin{equation}\label{eq67}
\left\{
\begin{array}{l}
P^{s}=f\left(\theta^{DF},c,\mu+\frac{\nu(g^{s,r}-g^{s,d})}{g^{r,d}},
\right.\\
\left.~~~~Q^s+\frac{Q^r(g^{s,r}-g^{s,d})}{g^{r,d}},g^{s,r},1+\frac{(g^{s,r}-g^{s,d})^2}{\left(g^{r,d}\right)^2}\right),\\
P^{r}=\frac{P^{s}(g^{s,r}-g^{s,d})}{g^{r,d}}.\\
\end{array}\right.\!\!\!\!
\end{equation}
\end{lem}
The proof of Lemma \ref{lem5} is provided in Appendix~\ref{distributedpowerallocation1}.

Finally, the channel allocation problem \eqref{eq13} can be decomposed as $T$ independent local channel allocation subproblems for each control node. The $t$th control node need to solve the following subproblem:
\begin{subequations}
\label{eq19}
\begin{align}
\!\!\!\!\max
\limits_{\theta_{m}^{DT},\theta_{mj}^{DF}\geq\theta_{\min}}&\sum_{\{m|c(m)=t\}}\left(R_m^{DT}+\sum_{j\in\mathcal J(m)} R_{mj}^{DF}\right)\\
\textrm{s.t.}~~~~~~&   \sum_{\{m|c(m)=t\}}\left(\theta_{m}^{DT}+\sum_{j\in\mathcal
J(m)}\theta_{mj}^{DF}\right) \leq \beta_t,\\\label{eq59}
&\theta_{m}^{DT},\theta_{mj}^{DF}\geq\theta_{\min}.
\end{align}
\end{subequations}

%

%
Both $R^{DT}$ and $R^{DF}$ are strictly concave in their corresponding channel resource variables $\theta^{DT}$ and $\theta^{DF}$, respectively. Therefore, we can use standard Lagrangian duality techniques to solve subproblem \eqref{eq19}. Let $\omega_t$ be the Lagrange multiplier for constraint \eqref{eq59}, the optimal value of $\theta^{DT}$ is given by
\begin{equation}\label{eq48}
\theta^{DT} =\left\{\!\!
\begin{array}{l}
\textrm{the root $x$ of \eqref{eq55} with $a,b$ given by \eqref{eq40}},\\
\textrm{~~~~~~~~~~~~~~~~~~~~~~~~~~~~~~~~~~if $x>\theta_{\min}$};\\
\theta_{\min}~~~~~~~~~~~~~~~~~~~~~~~~~~~,~\textrm{otherwise},
\end{array}
\right.\!\!\!\!\!\!\!\!\!
\end{equation}
where the root $x$ is determined by
\begin{align}
\label{eq55}
\log_2\left(1+\frac{a}{x}\right)-\frac{\frac{a}{x}}{\ln2\left(1+\frac{a}{x}\right)}=b.
\end{align}
and $a,b$ are given by
\begin{align}\label{eq40}
a=g^{s,d}_{m}{P}^s_{m},~b=\omega_t.
\end{align}
The optimal value of $\theta^{DF}$ is given by
\begin{equation}
\theta^{DT} =\left\{\!\!
\begin{array}{l}
\textrm{the root $x$ of \eqref{eq55} with $a,b$ given by \eqref{eq41}},\\
\textrm{~~~~~~~~~~~~~~~~~~~~~~~~~~~~~~~~~~if $x>\theta_{\min}$};\\
\theta_{\min}~~~~~~~~~~~~~~~~~~~~~~~~~~~,~\textrm{otherwise},
\end{array}
\right.\!\!\!\!\!\!\!\!\!
\end{equation}
where $a,b$ are determined as
\begin{align}\label{eq41}
a\!=\!2\min\left\{g^{s,r}_{mj}{P}^s_{mj},g^{s,d}_{m}{P}^s_{mj}+g^{r,d}_{mj}{P}^r_{mj}\right\},~b=2\omega_t.
\end{align}
If $\theta^{DT}$ or $\theta^{DF}$ equals to $\theta_{\min}$, the corresponding communication link is not admitted for transmission.

Equation \eqref{eq55} can be solved by the Newton's method with guaranteed global convergence, which typically requires only 5-8 iterations. The optimal value of dual variable $\omega_t$ is obtained through bisection method to satisfy \eqref{eq59} with equality. The total computation complexity to solve subproblem \eqref{eq19} is quite low.

In each iteration, Algorithm $\mathcal{A}$ only requires an exchange of the local resource allocation solution, dual variables $\mu_l$ and $\nu_j$, and auxiliary variables among the source, relay, destination, and control nodes of each DF relay link. Hence, Algorithm $\mathcal{A}$ is insensitive to the variations in the network size and topology.

The key reason for introducing the quadratic terms is to ensure that the subproblem \eqref{eq18} has a unique solution. Suppose that $c_m=c_{mj}=0$, then problem \eqref{eq7} reduces to the original problem $(\textrm{P})$, and the subproblem \eqref{eq18} may have many optimal solutions. For example, if $\nu=0$, $\nu/c$ in the expression of $P^{r}$ in \eqref{eq53} can be an arbitrary positive number in $[0,Q^r]$, hence the relay power solutions are non-unique. Since no global network information is available when solving the local power allocation subproblem \eqref{eq18}, it is quite difficult to find a global feasible solution among all the local optimal solutions. As a result, the power allocation variable will keep oscillating, even though the dual variable converges to the optimal solution. However, as will be shown in Theorem \ref{thm1}, we are able to overcome this convergence issue by using Algorithm $\mathcal{A}$.
\subsection{Convergence Analysis of Algorithm $\mathcal{A}$}
\begin{deff}
We say a point $(\vec{Q}^*,\vec{\nu}^*)$ is a \emph{stationary point} for given $\vec{\theta}$, if
\begin{align}
&\vec{Q}^*=\arg\max\limits_{\vec{x}\geq0}L\left(\vec{x},\vec{Q}^*;\vec{\nu}^*\right),\label{eq36}\\
&E\vec{Q}^*-\vec{P}_{\max}\leq0, ~\vec{\nu}^*\geq0,\label{eq72}\\
&\vec{\nu}^{*}\otimes\left(E\vec{Q}^*-\vec{P}_{\max}\right)=0,
\label{eq71}
\end{align}
where $\vec{x}\otimes \vec{y}$ represents the Hadamard (elementwise) product of two vecters $\vec{x}$ and $\vec{y}$ with the same dimension.
\end{deff}

The convergence and optimality of Algorithm $\mathcal{A}$ is established in the following theorem:
\begin{thm}
\label{thm1}
If $c_{m}$ and $c_{mj}$ are small enough, and the dual step-size $\alpha_l$ satisfies
\begin{align}
\max_l{\alpha_l}\leq\frac{1}{2 {S}}\min_{m,j}\{c_m,c_{mj}\},
\end{align}
where $S$ is the maximal number of links that a source or relay node can participate, given by
\begin{align}
S=&\max\!\left\{\!\max_l\!\left[\sum_{\{m|s(m)=l\}}\!\!\!\!(1+J(m))\right],\max_j\!\left[\!\sum_{\{m|j\in\mathcal J(m)\}}\!\!\!\!1\right]\!\right\}\!.
\end{align}
Then when the update period of the channel resources, i.e., $K$, is sufficiently large,
our proposed {Algorithm}~$\mathcal{A}$ converges to an optimal solution
to Problem $(\textrm{P})$.
\end{thm}

We note that Theorem \ref{thm1} only provides a sufficient condition for the convergence of Algorithm~$\mathcal{A}$.
One interesting future direction is to explore some weaker conditions to ensure the
convergence of Algorithm~$\mathcal{A}$.

The proof of Theorem \ref{thm1} is similar to that of Proposition 4 in \cite{Lin06}. However, since the form of the objective function in Problem $(P)$ is quite different from that of \cite{Lin06}, we require the following result to replace Lemma 3 in \cite{Lin06}:

\begin{lem}
\label{lem2}
For any given $\vec{Q}$ and $\vec{\theta}$, let $(\vec{Q}^*,\vec{\nu}^*)$ be a stationary point satisfying \eqref{eq36}-\eqref{eq71}, $(\vec{P}_1,\vec{\nu}_1)$ and
$(\vec{P}_2,\vec{\nu}_2)$ be the corresponding maximizers of the
Lagrangian \eqref{eq50}, i.e.,
$$\vec{P}_1=\arg\max\limits_{\vec{P}\geq0}L\left(\vec{P},\vec{Q},\vec{\theta};\vec{\nu}_1\right)
,~~\vec{P}_2=\arg\max\limits_{\vec{P}\geq0}L\left(\vec{P},\vec{Q},\vec{\theta};\vec{\nu}_2\right).$$ If
$c_{mj}$ is small enough, then
\begin{eqnarray}\label{eq78}
&&\left[\nabla  R\left(\vec{P}_1, \vec{\theta}\right)-\nabla  R\left(\vec{Q}^*, \vec{\theta}\right)\right]^T\left(\vec{P}_2-\vec{Q}^*\right)\nonumber\\
\leq\!\!\!\!\!\!\!\!&&\frac{1}{2}\left(\vec{\nu}_2-\vec{\nu}_1\right)^TEV^{-1}E^T\left(\vec{\nu}_2-\vec{\nu}_1\right),
\end{eqnarray}
where $\nabla R(\vec{P}, \vec{\theta})$ and
$\nabla R(\vec{Q}^*, \vec{\theta})$ are subgradients of $R$ with respect to $\vec{P}$ satisfying
\begin{align}
\label{eq21}
&\nabla  R\left(\vec{P}, \vec{\theta}\right)-E^T\vec{\nu}-V\left(\vec{P}-\vec{Q}\right)=0,\\
\label{eq22}
&\nabla  R\left(\vec{Q}^*, \vec{\theta}\right)-E^T\vec{\nu}^*=0.
\end{align}
\end{lem}
\begin{proof}
A sketch of the proof is as follows. Since the rate function $R_{mj}^{DF}$ in \eqref{eq1} has two different forms depending on whether $P^s_{mj}g^{s,r}_{mj}\geq P^s_{mj}g^{s,d}_{m}+P^r_{mj}g^{r,d}_{mj}$ or $P^s_{mj}g^{s,r}_{mj}>P^s_{mj}g^{s,d}_{m}+P^r_{mj}g^{r,d}_{mj}$, we need to prove
the inequality
\begin{eqnarray}
\!\!\!\!\!\!\!\!\!\!\!\!&&\Big[\nabla_{\vec{P}_{mj}}R_{mj}^{DF}(\vec{P}_{mj,1})-\nabla_{\vec{P}_{mj}}R_{mj}^{DF}(\vec{Q}_{mj}^*)\Big]^T\big(\vec{P}_{mj,2}-\vec{Q}_{mj}^*\big)\!\!\!\nonumber\\
\!\!\!\!\!\!\!\!\!\!\!\!&&\leq \frac{1}{2c_{mj}}\left[(\mu_{s(m),2}-\mu_{s(m),1})^2+(\nu_{j,2}-\nu_{j,1})^2\right], \forall m,j,
\end{eqnarray}
for 8 different cases corresponding to different values of $\vec{P}_{mj,1}$, $\vec{P}_{mj,2}$, and $\vec{Q}_{mj}^*$, where
$\vec{P}_{mj}=(P_{mj}^s,P_{mj}^r)$. For each case, we show that \eqref{eq78} holds, provided that $c_{mj}$ is small enough.
\ifreport
The detailed proof is very technical and is given in Appendix~\ref{proofoflemma}.
\else
The detailed proof is very technical and is relegated to \cite{Yin2011techreport}.
\fi
\end{proof}

With Lemma \ref{lem2}, we are ready to prove Theorem \ref{thm1}. For any given channel allocation $\vec{\theta}$, Proposition 4 in \cite{Lin06} guarantees that $\vec{Q}$ converges to the optimal power allocation solution in Step 1) and 2) of Algorithm~$\mathcal{A}$. The optimal channel allocation for given power allocation is obtained in Step 3) of Algorithm~$\mathcal{A}$. Therefore, Algorithm~$\mathcal{A}$ is a block coordinate optimization algorithm, which optimizes the power allocation and channel allocation iteratively. Then, Theorem \ref{thm1} follows from Proposition 2.7.1 in \cite{BK:Bertsekas}. The convergence with $K$ bounded away from infinity is empirically verified during our simulations.

\section{Centralized Channel Resource Adjustment}
\label{sec4}
In practice, the spatial distribution of wireless traffic is usually non-uniform and varies from time to time. The pre-assigned channel resources of the distributed control nodes may be inadequate to support the non-uniform traffic distribution.
In this section, we develop a centralized channel resource adjustment algorithm on top of the proposed distributed power allocation algorithm, so as to provide additional channel resources to the bottleneck control nodes and fulfill traffic load balance.

The overall algorithm addresses the following resource allocation problem, which has one extra constraint \eqref{eq49}:
%
\begin{subequations}\label{eq75}
\begin{align}
\!\!\!\!(\textrm{P1})\!\!\!\max
\limits_{\substack{P^s_{m},P^s_{mj},P^r_{mj},\\\theta_{m}^{DT},\theta_{mj}^{DF},\beta_t}}&\sum_{m=1
}^M\left(R_m^{DT}+\sum_{j\in\mathcal J(m)} R_{mj}^{DF}\right)\\
\rm{s.t.}~~~~&\!\!\!\!\!\!\!\!\sum_{\{m|s(m)=l\}}\!\!\left(P^s_{m}+\!\!\sum_{j\in\mathcal
J(m)} P^s_{mj}\right)\!\leq\!
P^s_{l,\max}, \forall l\!\!\!\! \\
&\!\!\!\!\!\!\!\!\sum_{\{m|j\in\mathcal J(m)\}}P^r_{mj} \leq
P^r_{j,\max}, \forall~j
\\
&
P^s_{m},P^s_{mj},P^r_{mj}\geq0,\forall~m,j\\
&\!\!\!\!\!\!\!\!\sum_{\{m|c(m)=t\}}\!\!\left(\theta_{m}^{DT}+\sum_{j\in\mathcal J(m)}\theta_{mj}^{DF}\right) \leq \beta_t,~ \forall t\\
&\theta_{m}^{DT},\theta_{mj}^{DF}\geq\theta_{\min},\forall~m,j.\\
&\vec{\beta}\in \mathscr{B}, \label{eq49}
\end{align}
\end{subequations}
where $\vec{\beta}$ is a $T$ dimensional vector representing the channel resource proportions $\beta_t$ pre-assigned to the control nodes and $\mathscr{B}$ is the feasible set containing all the possible choices of $\vec{\beta}$.
The feasible set $\mathscr{B}$ satisfies the following two principles:
\begin{enumerate}
\item Neighboring control nodes are assigned to orthogonal channels to suppress co-channel interference.
\item Distant control nodes are allowed to reuse the same channel.
\end{enumerate}
Under these principles, $\mathscr{B}$ can be expressed as a convex set bounded by linear constraints of $\vec{\beta}$, i.e., a polyhedron \cite{BK:BoydV04}. For instance, If control node $2$ is close to the control nodes $1$ and $3$, while control node $1$ are relative far from control node $3$, the channel resource constraint set $\mathscr{B}$ is given by
\begin{eqnarray}
\mathscr{B}=\left\{(\beta_1,\beta_2,\beta_3)
|\beta_1+\beta_2\leq 1,\beta_2+\beta_3\leq 1,\beta_1,\beta_2,\beta_3\geq 0\right\}\!.\!\!\!\!\!\!\!\!\!\!\!\!\!
\end{eqnarray}
Therefore, problem $(\textrm{P1})$ is a convex optimization problem.

Problem $(\textrm{P}1)$ can be reformulated as
\begin{subequations}
\label{eq69}
\begin{align}
\max
\limits_{\beta_t} &~R^\star(\vec{\beta})\\
\rm{s.t.}& ~\vec{\beta}\in \mathscr{B},
\end{align}
\end{subequations}
\noindent where $R^\star(\vec{\beta})$ is the optimal value of Problem $(\textrm{P})$.
Problem \eqref{eq69} is solved by using the subgradient method \cite{BK:Bertsekas,Subgradient_Boyd}. The subgradient $\frac{\partial R^\star(\vec{\beta})}{\partial \beta_t}$ is exactly the optimal dual variable $\omega_t^\star$ associated with the constraint~\eqref{eq15} \cite{BK:Bertsekas,Palomar06,Palomar07}, which can be obtained for free during the process of solving the subproblem \eqref{eq19}. The outer-layer subgradient algorithm operates on a much slower time scale, so that each subgradient step is performed after an approximate solution to Problem $(\textrm{P})$ has been derived. At the $q$th iteration, the subgradient method updates $\beta_t$ by \cite{Subgradient_Boyd}
\begin{align}
\label{eq70}
\beta_t(q+1)=\left[\beta_t(q)+\delta(q)\omega_t^\star\right]_{\mathscr{B}},
\end{align}
where $\delta(q)$ is the step-size of the $l$th subgradient update and $\left[x\right]_{\mathscr{B}}$ is the projection of $x$ onto the set $\mathscr{B}$.
If the step size $\delta(q)$ is chosen according to a diminishing step size rule, the subgradient algorithm converges to the optimal solution to \eqref{eq75} \cite{Subgradient_Boyd}. If the step size $\delta(q)$ is a constant value, the subgradient algorithm converges to a neighborhood of the optimal solution \cite{Subgradient_Boyd}. The convergence speed of the subgradient method can be improved if one further considers the acceleration techniques in
\cite{BK:Bertsekas,BK:Bazaraa06,BK:Shor85}. The overall resource allocation algorithm for solving problem $(\textrm{P}1)$ is summarized in Algorithm \ref{alg1}.
\begin{algorithm}
  \caption{The resource allocation algorithm for solving problem $(\textrm{P}1)$}\label{alg1}
\begin{algorithmic}[1]
  \State Set the iteration number $q= 1$; initialize the dual variable $\beta_t(1)$ for all $t\in\mathcal {T}$.
  \State Solve problem $(\textrm{P})$ by our proposed distributed power allocation Algorithm $\mathcal {A}$ and obtain $\omega_t^\star$.
  \State Update $\beta_t$ by \eqref{eq70} and return to Step 2, until $\vec{\beta}$ has converged.
\end{algorithmic}
\end{algorithm}

\begin{figure}[!t] \centering
    \resizebox{0.45\textwidth}{!}{
    \includegraphics{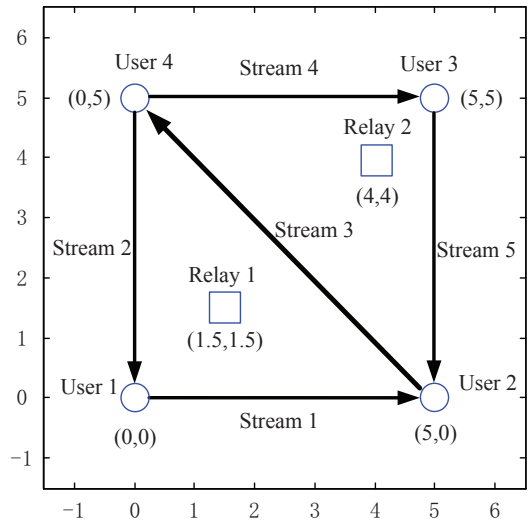}}
    \caption{The topology of the considered DF relay network.}
    \label{fig3}
\end{figure}

\section{Discussions}\label{sec6}
While the focus of this paper is on two-hop DF relay networks, our resource allocation solution can easily be extended to multi-hop DF relay networks, which employ single-hop point-to-point transmission and two-hop relaying as the basic wireless link between two nodes. Then, the cross-layer optimization of multi-hop DF relay networks, including congestion control, routing, and resource allocation, can be resolved by following the techniques in \cite{LinJsac06}. Since the congestion control and routing solutions are quite similar with those of traditional multi-hop wireless networks \cite{LinJsac06}, we focus on the more difficult resource allocation component in this paper.

\section{Simulation Results} \label{sec5}
This section provides some simulation results to examine the performance of our proposed resource allocation algorithm. We consider a DF relay network with 4 source/destination user nodes and 2 relay nodes. The topology of the network is illustrated in Fig. \ref{fig3}. There are 5 data streams in this network, which can transmit either through source-destination DT links or by the assistance of the 2 relay nodes. We assume that User 1 and the 2 relay nodes are 3 control nodes of the network. Data streams 1 and 2 are managed by User 1, data stream 3 is managed by Relay 1, and data streams 4 and 5 are managed by Relay 2. These control nodes are assigned with orthogonal channel resources, and the feasible set $\mathscr{B}$ is determined by
\begin{eqnarray}
\mathscr{B}=\left\{(\beta_1,\beta_2,\beta_3)| \beta_1\!+\!\beta_2\!+\!\beta_3\!\leq\! 1,\beta_1,\beta_2,\beta_3\geq 0\right\}.\!
\end{eqnarray}
The channel power gain between two nodes is determined by a large-scale path loss component with a path loss factor of 4. Each source and relay node has the same amount of transmission power, i.e. $P_{l,\max}^s = P_{j,\max}^r = P_{\max}$ for all $l,j$. The received signal-to-noise ratio (SNR) at unit distance from a transmitting node is $\frac{P_{\max}}{N_0W} =25$dB.

\begin{table}
\renewcommand{\arraystretch}{1.3}
\caption{Spectrum efficiency (bits/s/Hz) and relay selection of each data stream.}
  \centering
  \label{tab3}
    \begin{tabular}{c c c c c c}
      \hline\hline
         &Stream~1 & Stream~2 & Stream~3 & Stream~4 & Stream~5 \\ \hline
        with DF relaying & 0.72 & 0.41 & 0.10 & 0.05 & 0.63  \\
        without DF relaying & 0.43 & 0.22 & 0.11 & 0.22 & 0.43 \\
        \hline
        Relay selection & Relay~1 & Relay~1 & Relay~1~\&~2 & Relay~2 & Relay 2\\
      \hline\hline
    \end{tabular}
\end{table}

\begin{figure}[!t] \centering
    \resizebox{0.45\textwidth}{!}{
    \includegraphics{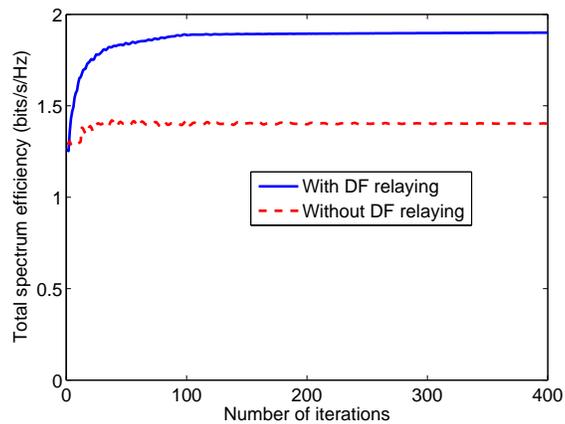}}
    \caption{The evolution of the total spectrum efficiency achieved by Algorithm \ref{fig1}.}
    \label{fig6}
\end{figure}

\begin{figure}[!t] \centering
    \resizebox{0.45\textwidth}{!}{
    \includegraphics{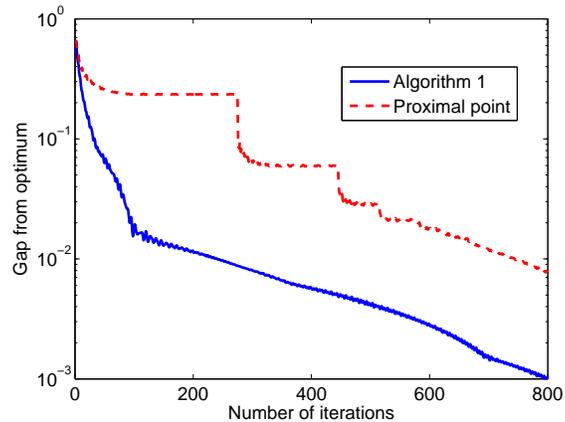}}
    \caption{The convergence performance of Algorithm \ref{fig1} and traditional proximal point method.}
    \label{fig7}
\end{figure}

\begin{figure}[!t] \centering
    \resizebox{0.45\textwidth}{!}{
    \includegraphics{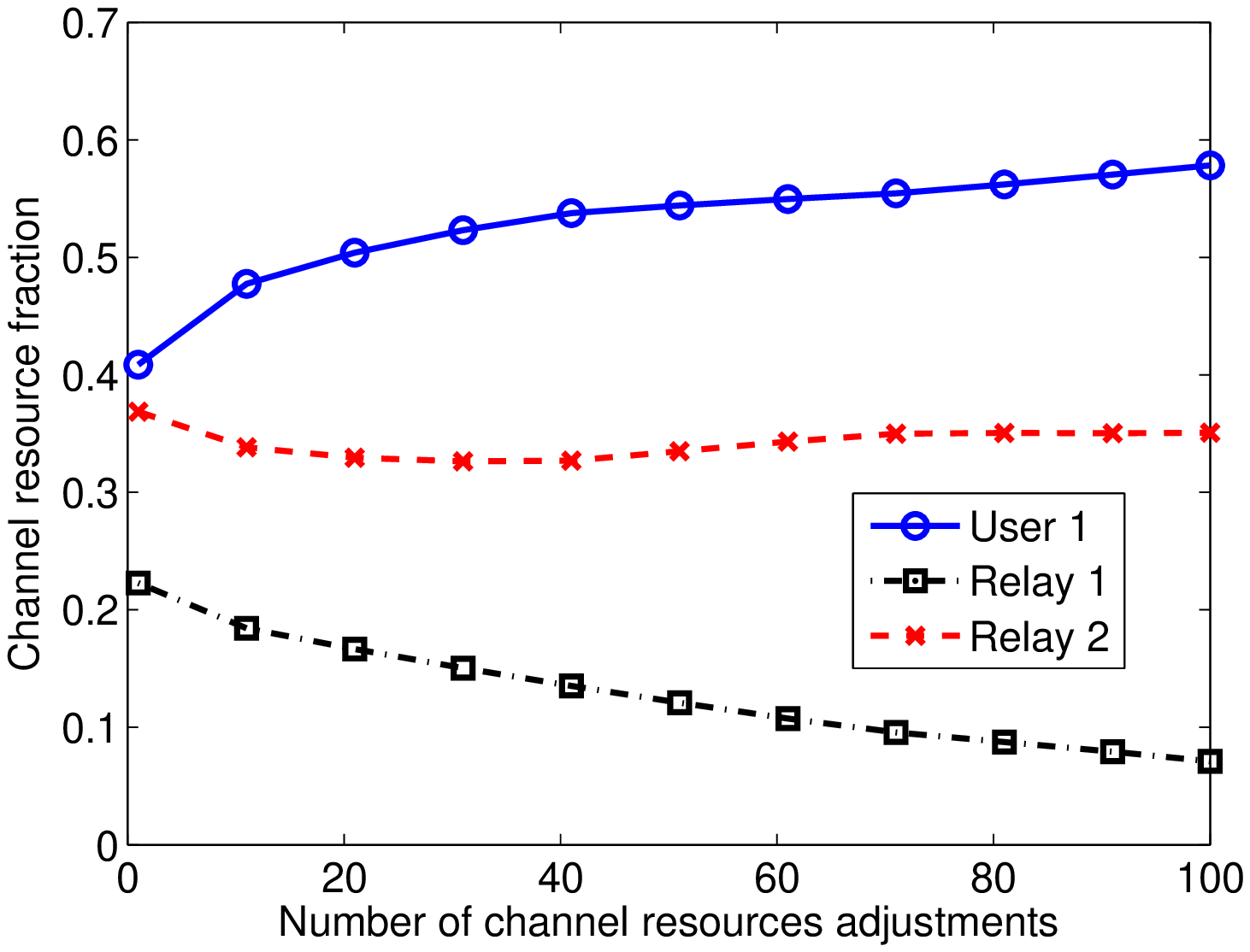}}
    \caption{The evolution of channel resource adjustment of the distributed control nodes with DF relaying.}
    \label{fig8}
\end{figure}

In Algorithm \ref{alg1}, the channel allocation $\vec{\theta}$ is computed according to \eqref{eq19} after every $K=2$ iterations, and the channel resources of the control nodes $\vec{\beta}$ is adjusted according to \eqref{eq70} after every $10$ iterations. The algorithm parameters are chosen as $c_m = c_{mj} = 10^{-4}$, $\alpha_l = 5\times10^{-5}$, $\delta(q) = 0.3$, and $\theta_{\min}=0.01$. Our numerical experience suggests that there is a wide range of the algorithm parameters $c_m$ and $c_{mj}$ that ensures the convergence and optimality of Algorithm \ref{alg1}.

The spectrum efficiency and relay selection of each data stream are provided in Tab. \ref{tab3}. Figure \ref{fig6} illustrates the evolution of our resource allocation Algorithm \ref{alg1}. The total spectrum efficiency of the network increases 35.5\% by deploying DF relay techniques. The convergence speeds of Algorithm \ref{alg1} and traditional proximal point algorithm \cite{Bertsekas89} are shown in Fig. \ref{fig7}.  In the proximal point algorithm, the auxiliary update of \eqref{eq23} is implemented only after the updates of \eqref{eq14} and \eqref{eq20} converge and satisfy the criteria $\|\vec{\nu}(k+1)-\vec{\nu}(k)\|_2\leq 0.01 \|\vec{\nu}(k)\|_2$. Since Algorithm \ref{alg1} has less layers of iterations than the proximal point algorithm, it has a much faster convergence speed.
We note that this is the convergence speed when the algorithm is cold started. In practice, since the channel condition varies slowly, the resource allocation solution from the previous run of the algorithm is an excellent initial state. By using the previous solution for warm-starting the algorithm, the algorithm converges much faster.

The evolution of the channel resource adjustment of the distributed control nodes, i.e., User 1, Relay 1, and Relay 2, is illustrated in Fig. \ref{fig8} for the case that DF relay technique is employed, and in Fig. \ref{fig9} for the case that the relay nodes only act as control nodes but not perform DF relaying. 

%

\begin{figure}[!t] \centering
    \resizebox{0.45\textwidth}{!}{
    \includegraphics{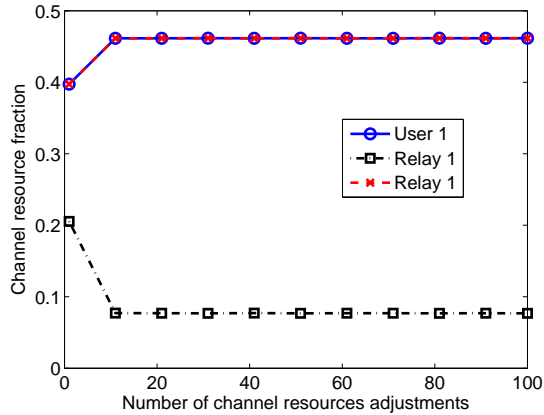}}
    \caption{The evolution of channel resource adjustment of the distributed control nodes without DF relaying.}
    \label{fig9}
\end{figure}

\section{Conclusion}\label{conclusion}
In this paper, we have developed a distributed resource allocation algorithm
to jointly optimize the power allocation, channel allocation and relay selection for DF relay networks with a large number of sources, relays, and destinations. Conventional dual decomposition techniques are not applicable, because the objective function is not strictly concave. We resolve this problem by adding some quadratic terms to make the objective function strictly concave. In this algorithm, information exchange only occurs locally among the source, relay, destination, and control nodes of each DF relay link. We have establish the convergence and optimality of our distributed resource allocation algorithm. A centralized channel
resource adjustment algorithm has been developed on top of the distributive resource allocation algorithm to achieve traffic load balance. Numerical results are provided to illustrate the benefits of our proposed algorithm.

\bibliographystyle{IEEEtran}
\bibliography{refs}

\appendix


\subsection{Proof of Lemma \ref{lem4}}
\label{distributedpowerallocation}
The Karush-Kuhn-Tucker (KKT) conditions \cite{BK:BoydV04} of \eqref{eq17} imply
\begin{align}
\label{eq300}
&\frac{g^{s,d}}{\ln2 (1\!+\!\frac{g^{s,d} P^s}{\theta^{DT}})}\!-\!{c}\left(P^s-Q^{s}\right)\!-\!\mu
\!\left\{\!\!\!\begin{array}{l}=0,\textrm{if}~P^{s}\!>\!0\\\leq 0,\textrm{if}~P^{s}\!=\!0\end{array}\right.\!\!\!.\!\!\!
\end{align}
When $P^{s}>0$, \eqref{eq300} achieves equality, $P^{s}$ is thus the positive root of a quadratic equation equivalent with \eqref{eq300}; otherwise, $P^{s}_{}=0$. Summarizing these two cases, the value of $P^{s}_{}$ is given by \eqref{eq52}.

\subsection{Proof of Lemma \ref{lem5}}
\label{distributedpowerallocation1}
Let us define the rate functions
\begin{align}\label{eq46}
R_1 =
&\frac{\theta^{DF}}{2}\log_2\left[1+\frac{2\left(P^s_{}g^{s,d}_{}+P^r_{}g^{r,d}_{}\right)}{\theta^{DF}}\right],\\
R_2
=
&\frac{\theta^{DF}}{2}\log_2\left(1+\frac{2P^s_{}g^{s,r}_{}}{\theta^{DF}}\right).\label{eq47}
\end{align}
Therefore, the achievable rate of DF relaying can be expressed as $R^{DF}=\min \left\{R_1,R_2\right\}.$ Problem \eqref{eq18} is equivalent to the following problem:
\begin{subequations}
\label{eq32}
\begin{align}
\!\!\!\!\max_{\substack{t,P^s_{},P^r_{}\geq0}}&
t\!-\!\frac{c_{}}{2}\left(P^s_{}\!-\!Q^{s}_{}\right)^2\!-\!\frac{c_{}}{2}\left(P^r_{}
\!-\!Q^{r}_{}\right)^2\!-\!\mu_{}P^s_{}\!-\!\nu_{}P^r_{}\\
\textrm{s.t.}~~~ & t\leq R_1,~~t\leq R_2,
\end{align}
\end{subequations}
\noindent where $R_1$ and $R_2$ are defined in \eqref{eq46} and \eqref{eq47}.
\noindent The Lagrangian of problem \eqref{eq32} is
\begin{align}\label{eq76}
&L(t,P^s_{},P^r_{}; \tau,\zeta,\mu_{},\nu_{})\nonumber\\
=&t+\tau(R_1-t)+\zeta(R_2-t) -\frac{c_{}}{2}\left(P^s_{}-Q^{s}_{}\right)^2\nonumber\\
&-\frac{c_{}}{2}\left(P^r_{}-Q^{r}_{}\right)^2-\mu_{}P^s_{}-\nu_{}P^r_{}.
\end{align}
The KKT optimality conditions of problem \eqref{eq32} indicate
\begin{align}\label{eq42}
&\frac{\partial L}{\partial t} = 1-\tau-\zeta=0,\\
&\tau\geq0,~R_1-t\geq0,~\tau(R_1-t)=0,\\
&\zeta\geq0,~R_2-t\geq0,~\zeta(R_2-t)=0.\label{eq43}
\end{align}
By \eqref{eq42}, the Lagrangian \eqref{eq76} can be simplified as
\begin{align}
&L(P^s_{},P^r_{}; \tau,\mu_{},\nu_{})\nonumber\\
=&\tau R_1+(1-\tau)R_2 -\frac{c_{}}{2}\left(P^s_{}-Q^{s}_{}\right)^2-\frac{c_{}}{2}\left(P^r_{}-Q^{r}_{}\right)^2\nonumber\\
&-\mu_{}P^s_{}-\nu_{}P^r_{}.
\end{align}
Moreover, from the KKT optimality conditions \eqref{eq42}-\eqref{eq43},
we obtain
\begin{align}
\label{eq35}
\left\{
\begin{array}{l}\textrm{if}~~~R_1^\star> R_2^\star,~~~\tau=0 ;\\
\textrm{if}~~~R_1^\star< R_2^\star,~~~\tau=1 ;\\
\textrm{if}~~~R_1^\star= R_2^\star,~~0\leq\tau\leq1,
\end{array}\right.
\end{align}
where $R_1^\star$ and $R_2^\star$ are corresponding rate values \eqref{eq46} and \eqref{eq47} at the optimal power allocation solution.
Therefore, problem \eqref{eq32} can be address by solving the Lagrangian maximization problem
\begin{align}\label{eq51}
\max_{P^s_{},P^r_{}\geq0}L(P^s_{},P^r_{}; \tau,\mu_{},\nu_{}),
\end{align}
for the 3 cases expressed in \eqref{eq35}.
%
%

\textbf{Case~1}:
When $R_1^\star> R_2^\star,$ we have $\tau=0$. The KKT conditions of \eqref{eq51} are
\begin{align}
\label{eq301}
&\frac{g^{s,r}}
{\ln2(1\!+\!\frac{2g^{s,r}{P}^s_{}}{{\theta}^{DF}})}\!-\!\mu_{}\!
+\!c({P}^s\!-\!Q^s)\!\left\{\!\!\!\begin{array}{l}=0,\textrm{if}~ P^{s}_{}\!>\!0\\
\leq 0,\textrm{if}~P^{s}_{}\!=\!0\end{array}\right.\!\!\!,\!\!\!\\
&-\nu_{}-c({P}^r-Q^r)\left\{\begin{array}{l}=0,\textrm{if}~P^{r}_{}>0\\\leq 0,\textrm{if}~P^{r}_{}=0\end{array}\right..\label{eq302}
\end{align}
The solution to~\eqref{eq301} is similar to that of ~\eqref{eq300}. Further considering \eqref{eq302}, the optimal source and relay power is given by \eqref{eq53}. Note that Case 1 requires $R_1^\star> R_2^\star$, which can be equivalently expressed by $g^{r,d}P^r>(g^{s,r}-g^{s,d})P^s$.

\textbf{Case~2}:
When $R_1^\star< R_2^\star,$ $\tau = 1$. The KKT conditions of ~\eqref{eq51} are given by
\begin{align}
\label{eq303}
&\frac{g^{s,d}}
{\ln2(1+\frac{2g^{s,d}{P}^s+2g^{r,d}{P}^r}{\theta^{DF}})}-\mu_{}-c({P}^s-Q^s)\nonumber\\
&~~~~~~~~~~~~~~~~~~~~~~~~~~~~~~~~~~~~\left\{\begin{array}{l}=0, ~\textrm{if}~P^{s}_{}>0\\\leq 0,~\textrm{if}~P^{s}_{}=0\end{array}\right.,\\
&\frac{g^{r,d}}
{\ln2(1+\frac{2g^{s,d}{P}^s+2g^{r,d}{P}^r}{\theta^{DF}})}-\nu_{}-c({P}^r-Q^r)\nonumber\\
&~~~~~~~~~~~~~~~~~~~~~~~~~~~~~~~~~~~~\left\{\begin{array}{l}=0, ~\textrm{if}~P^{r}_{}>0\\\leq 0,~\textrm{if}~P^{r}_{}=0\end{array}\right..\label{eq304}
\end{align}

If $P^s>0,P^r>0$, then \eqref{eq303} and \eqref{eq304} take equality. By viewing $e=g^{s,d}{P}^s+g^{r,d}{P}^r$ as a whole body, we can get a quadratic equation of $e$ from \eqref{eq303} and \eqref{eq304}, which has a positive root given by \eqref{eq73}.
Moreover, by comparing \eqref{eq303} and \eqref{eq304} with equality, we obtain
\begin{align}\label{eq44}
\left[\mu_{}+c({P}^s-Q^s)\right]g^{r,d}=\left[\nu_{}+c({P}^r-Q^r)\right]g^{s,d}.
\end{align}
Substituting \eqref{eq44} into \eqref{eq73}, the optimal power allocation solution is derived as in \eqref{eq63}.

If $P^r=0$, \eqref{eq303} reduces to a formula similar with \eqref{eq301}, and its solution is given by \eqref{eq54}.

Note that Case 2 requires $R_1^\star< R_2^\star$, which is guaranteed by $g^{r,d}P^r<(g^{s,r}-g^{s,d})P^s$. The case of $P^s=0$ and $P^r>0$ cannot happen, since it violates the condition $R_1^\star< R_2^\star$.

\textbf{Case~3}:
If $R_1^\star= R_2^\star$, the KKT conditions of \eqref{eq51} are given by
\begin{align}
\label{eq305}
&\frac{\tau g^{s,d}}{\ln2(1\!+\!\frac{2g^{s,d}{P}^s+\!2g^{r,d}{P}^r}{\theta^{DF}})}\!+\!\frac{(1-\tau)g^{s,r}}{\ln2(1\!+\frac{2g^{s,r}{P}^s_{}}{{\theta}^{DF}})} \nonumber\\
&~~~~~~~~~~~~~~~~-\mu-c({P}^s-Q^s)\left\{\begin{array}{l}=0,~\textrm{if}~P^{s}_{}>0\\\leq 0,~\textrm{if}~
P^{s}_{}=0\end{array}\right.,\\
\label{eq306}
&\frac{\tau g^{r,d}}{\ln2(1+\frac{2g^{s,d}{P}^s+2g^{r,d}{P}^r}{\theta^{DF}})}
\nonumber\\
&~~~~~~~~~~~~~~~~-\nu-c({P}^r-Q^r)\left\{\begin{array}{l}=0,~\textrm{if}~P^{r}_{}>0\\\leq 0,~\textrm{if}~ P^{r}_{}=0\end{array}\right..
\end{align}
Since $R_1^\star= R_2^\star$, we have
\begin{align}
\label{eq307}
g^{r,d}P^r=(g^{s,r}-g^{s,d})P^s.
\end{align}
If $P^s>0,P^r>0$, both \eqref{eq305} and \eqref{eq306} achieves equality. By substituting \eqref{eq307} into \eqref{eq305} and \eqref{eq306}, we can eliminate $\tau$ and derive the optimal values of $P^s$ and $P^r$. Otherwise, $P^s=P^r=0$. These two cases are summarized in \eqref{eq67}.

\ifreport
\subsection{Proof of Lemma~\ref{lem2}}
\label{proofoflemma}
We proceed to show the inequalities
\begin{eqnarray}\label{eq80}
&&\Big[\nabla_{{P}^s_{m}}R_{m}^{DT}({P}^s_{m,1})-\nabla_{{P}^s_{m}}R_{m}^{DT}({Q}_{m}^{s*})\Big]^T\big({P}^s_{m,2}-{Q}_{m}^{s*}\big)\nonumber\\
&&\leq \frac{1}{2c_{m}}(\mu_{s(m),2}-\mu_{s(m),1})^2, \forall m,
\end{eqnarray}
and
\begin{eqnarray}\label{eq81}
\!\!\!\!\!\!\!\!\!\!\!\!&&\Big[\nabla_{\vec{P}_{mj}}R_{mj}^{DF}(\vec{P}_{mj,1})-\nabla_{\vec{P}_{mj}}R_{mj}^{DF}(\vec{Q}_{mj}^*)\Big]^T\big(\vec{P}_{mj,2}-\vec{Q}_{mj}^*\big)\!\!\!\nonumber\\
\!\!\!\!\!\!\!\!\!\!\!\!&&\leq \frac{1}{2c_{mj}}\left[(\mu_{s(m),2}-\mu_{s(m),1})^2+(\nu_{j,2}-\nu_{j,1})^2\right], \forall m,j,
\end{eqnarray}
where ${P}^s_{m,i}$ is the maximizer of \eqref{eq17} corresponding to the multiplier $\mu_{s(m),i}$ for $i\in\{1,2\}$, and $\vec{P}_{mj,i}=(P_{mj,i}^s,P_{mj,i}^r)$
is the maximizer of \eqref{eq18} corresponding the multiplier
$(\mu_{s(m),i},\nu_{j,i})$. The asserted result \eqref{eq78} follows, if we take the summation of the inequalities \eqref{eq80} and \eqref{eq81} for all the possible choices of $m$ and $j$.
Since $R_m^{DT}$ is a concave function of a single variable $P_m^s$, the techniques of \cite{Lin06} can be directly used here to prove \eqref{eq80}. In the sequel, we will show \eqref{eq81} for each DF relay link. Since we only need to focus on one DF relay link, the subscripts $m,j,s(m)$ are omitted in the sequel to facilitate our expressions.

Let us associate Lagrange multipliers
$L^s\geq0$ and $L^r\geq0$ for the constraints $P^s\ge0$ and $P^r\ge0$,
respectively, in the maximization of \eqref{eq18}. Using the Karush-Kuhn-Tucker condition, we can conclude that there must exist a subgradient $(\frac{\partial R^{DF}(\vec{P})}{\partial
P^s}, \frac{\partial R^{DF}(\vec{P})}{\partial
P^r})$ of $R^{DF}$ such that
\begin{align}
\label{equ:L1} &\frac{\partial R^{DF}(\vec{P})}{\partial
P^s}-\mu-c(P^s-Q^s)+L^s=0,\\
\label{equ:L2} &\frac{\partial R^{DF}(\vec{P})}{\partial
P^r}-\nu-c(P^r-Q^r)+L^r=0,\\
\label{equ:L3}
&L^sP^s=0,~~~L^rP^r=0,
\end{align}
where $\vec{P}= (P_{}^s,P_{}^r)$ represents the source and relay power of the considered DF relay link.
From \eqref{eq21} and \eqref{eq22}, we also have
\begin{align}\label{eq82}
&\nabla_{{P}^s} R^{DF}(\vec{P})-\mu-c(P^s-Q^s)=0,\\
&\nabla_{{P}^r} R^{DF}(\vec{P})-\nu-c(P^r-Q^r)=0. \label{eq83}
\end{align}
Comparing \eqref{equ:L1} and \eqref{equ:L2} with \eqref{eq82} and \eqref{eq83}, we see that
\begin{align*}
&\nabla_{{P}^s} R^{DF}(\vec{P}) = \frac{\partial R^{DF}(\vec{P})}{\partial
P^s}+ L^s,\\
&\nabla_{{P}^r} R^{DF}(\vec{P}) = \frac{\partial R^{DF}(\vec{P})}{\partial
P^r}+ L^r.
\end{align*}
Let $\vec{Q}^*=(Q^{s*},Q^{r*})$ be the source and relay power at the stationary point. Similarly, we can obtain
\begin{align*}
&\nabla_{{P}^s} R^{DF}(\vec{Q}^*) = \frac{\partial R^{DF}(\vec{Q}^*)}{\partial
P^s}+ L^{s*},\\
&\nabla_{{P}^r} R^{DF}(\vec{Q}^*) = \frac{\partial R^{DF}(\vec{Q}^*)}{\partial
P^r}+ L^{r*}.
\end{align*}

Then, we further have
\begin{align*}
    \!\!\!\label{equ:p1} &\Big[\nabla
    R^{DF}(\vec{P}_{1})-\nabla R^{DF}(\vec{Q}^*)\Big]^T\big(\vec{P}_{2}-\vec{Q}^*\big)\nonumber\\
    \!\!\!=&\Big[\frac{\partial R^{DF}(\vec{P}_{1})}{\partial
P^s}-\frac{\partial R^{DF}(\vec{Q}^*)}{\partial
P^s}\Big](P_{2}^s-Q^{s*})\nonumber\\
\!\!\!\!&+\Big[\frac{\partial R^{DF}(\vec{P}_{1})}{\partial
P^r}-\frac{\partial R^{DF}(\vec{Q}^*)}{\partial
P^r}\Big](P_{2}^r-Q^{r*})\\
\!\!\!\!&+\big(L_{1}^s\!-\!L^{s*}\big)\big(P_{2}^s\!-\!Q^{s*}\big)
+\big(L_{1}^r\!-\!L^{r*}\big)\big(P_{2}^r\!-\!Q^{r*}\big).
\end{align*}
We can use the arguments in \cite{Lin06} to show that
\begin{align*}
&\big(L_{1}^s\!-\!L^{s*}\big)\big(P_{2}^s\!-\!Q^{s*}\big)+\big(L_{1}^r\!-\!L^{r*}\big)\big(P_{2}^r\!-\!Q^{r*}\big)\\
\leq&\frac{1}{4c}\left[(\mu_{2}-\mu_{1})^2+(\nu_{2}-\nu_{1})^2\right].
\end{align*}
Now we only need to show
\begin{eqnarray}\label{eq84}
&\Big[\frac{\partial R^{DF}(\vec{P}_{1})}{\partial
P^s}-\frac{\partial R^{DF}(\vec{Q}^*)}{\partial
P^s}\Big](P_{2}^s-Q^{s*})\nonumber\\
\!\!\!\!+&\Big[\frac{\partial R^{DF}(\vec{P}_{1})}{\partial
P^r}-\frac{\partial R^{DF}(\vec{Q}^*)}{\partial
P^r}\Big](P_{2}^r-Q^{r*})\nonumber\\
\leq&\frac{1}{4c}\left[(\mu_{2}-\mu_{1})^2+(\nu_{2}-\nu_{1})^2\right].
\end{eqnarray}
for \eqref{eq81} to hold.

Let us further define
\begin{align*}
a_1^s=&\frac{\partial R^{DF}(\vec{P}_{1})}{\partial
P^s}-\frac{\partial R^{DF}(\vec{Q}^*)}{\partial
P^s},~~b_1^s=P_{1}^s-Q^{s*},\\
a_1^r=&\frac{\partial R^{DF}(\vec{P}_{1})}{\partial
P^r}-\frac{\partial R^{DF}(\vec{Q}^*)}{\partial
P^r},~~b_1^r=P_{1}^r-Q^{r*},\\
a_2^s=&\frac{\partial R^{DF}(\vec{P}_{2})}{\partial
P^s}-\frac{\partial R^{DF}(\vec{Q}^*)}{\partial
P^s},~~b_2^s=P_{2}^s-Q^{s*}\\
a_2^r=&\frac{\partial R^{DF}(\vec{P}_{2})}{\partial
P^r}-\frac{\partial R^{DF}(\vec{Q}^*)}{\partial
P^r},~~b_2^r=P_{2}^r-Q^{r*}.
\end{align*}
For ease of notation, let us define
$\gamma_{s}\triangleq-\frac{L_{2}^s-L_{1}^s}{c(P_{2}^s-P_{1}^s)}\ge0$,
$\gamma_{r}\triangleq-\frac{L_{2}^r-L_{1}^r}{c(P_{2}^r-P_{1}^r)}\ge0$,
$\gamma_{s}^0\triangleq1+\gamma_{s}$, and
$\gamma_{r}^0\triangleq1+\gamma_{r}$. Then, according to \eqref{equ:L1} and \eqref{equ:L2}, we have
\begin{align}
\mu_{2}-\mu_{1}&=a_2^s-a_1^s+c(b_1^s-b_2^s)+L_{2}^s-L_{1}^s\\
&=\big(a_1^s-a_2^s\big)-c\gamma_{s}^0\big(b_1^s-b_2^s\big),\\
\nu_{2}-\nu_{1}&=a_2^r-a_1^r+c(b_1^r-b_2^r)+L_{2}^r-L_{1}^r\\
&=\big(a_1^r-a_2^r\big)-c\gamma_{s}^0\big(b_1^r-b_2^r\big).
\end{align}
Moreover, (\ref{eq84}) can be equivalently expressed as
\begin{equation}
  \label{equ:equiv}
  a_1^sb_2^s+a_1^rb_2^r\le\frac{1}{4c}\Big[(\mu_{2}-\mu_{1})^2+(\nu_{2}-\nu_{1})^2\Big].
\end{equation}

The concavity of $R^{DF}$ in $(P^s,P^r)$ suggests
\begin{align}\label{eq61}
a_1^sb_1^s+a_1^rb_1^r\le0,~a_2^sb_2^s+a_2^rb_2^r\le0.
\end{align}
If $a_i^sb_i^s\le0$ and $a_i^rb_i^r\le0$ ($i=1,2$), one can show
Lemma~\ref{lem2} according to the arguments in \cite{Lin06}.
However, rather than having $a_i^sb_i^s\le0$ or $a_i^rb_i^r\le0$, we
only have \eqref{eq61}. In order to handle the
difficulty, we will discuss case by case to fully explore the
structure of $R^{DF}$ defined in \eqref{eq1}. In particular, we proceed the remaining proof by breaking into three
levels of cases:
\begin{itemize}
\item [1)] Break into \textbf{Case~1-7} based on all combinations of the signs of $a_1^sb_1^s$,
$a_1^rb_1^r$, $a_2^sb_2^s$, and $a_2^rb_2^r$.\\
\item [2)] In some cases, further break into subcases (I)-(IV) based on all
combinations of the signs of $a_1^sb_2^s$ and $a_1^rb_2^r$.\\
\item [3)] In some subcases, further break into mini-cases (1)-(8) based on
all combinations of $\vec{P}_{1}$,
$\vec{P}_{2}$, and $\vec{Q}^*$ lying in region $\textcircled{1}$ or region $\textcircled{2}$, where region $\textcircled{1}$ is defined as $R^{DF}=R_1$, i.e., $g^{s,d}P^s+g^{r,d}P^r\le
g^{s,r}P^s$, and region $\textcircled{2}$ is defined as $R^{DF}=R_2$, i.e., $g^{s,d}P^s+g^{r,d}P^r\ge
g^{s,r}P^s$.
\end{itemize}


We will explain in detail how to prove \eqref{equ:equiv} for each subcase and
mini-case in \textbf{Case~1}, and also in relative detail for \textbf{Case~2} to show that the proofs in \textbf{Case~1} and \textbf{Case~2} have similar logic. Since the techniques in \textbf{Case~3-7} are quite similar with that used in \textbf{Case~1}, we will omit most of the similar steps without repeating the same proof logic.

\noindent \textbf{Case~1}:
When $a_1^sb_1^s\ge0,~a_1^rb_1^r\le0,~a_2^sb_2^s\le0,~a_2^rb_2^r\ge0$.\\
(I) We first assume $a_1^sb_2^s\ge0,~a_1^rb_2^r\ge0$, then we have
$a_1^sa_2^s\le0,~b_1^sb_2^s\ge0,~a_1^ra_2^r\ge0,~b_1^rb_2^r\le0$.
Now we further break into mini-cases:\\
(1) $\vec{P}_{1}$ is in region
$\textcircled{1}$, $\vec{P}_{2}$ is in region
$\textcircled{1}$, $\vec{Q}^*$ is in region
$\textcircled{1}$. Then,
\begin{align*}
a_1^s=&\frac{g^{s,d}}{1+\frac{2}{\theta}(g^{s,d}P_{1}^s+g^{r,d}P_{1}^r)}-\frac{g^{s,d}}{1+\frac{2}{\theta}(g^{s,d}Q^{s*}+g^{r,d}Q^{r*})},\\
a_1^r=&\frac{g^{r,d}}{1+\frac{2}{\theta}(g^{s,d}P_{1}^s+g^{r,d}P_{1}^r)}-\frac{g^{r,d}}{1+\frac{2}{\theta}(g^{s,d}Q^{s*}+g^{r,d}Q^{r*})},\\
b_1^s=&P_{1}^s-Q^{s*},~~~b_1^r=P_{1}^r-Q^{r*},\\
a_2^s=&\frac{g^{s,d}}{1+\frac{2}{\theta}(g^{s,d}P_{2}^s+g^{r,d}P_{2}^r)}-\frac{g^{s,d}}{1+\frac{2}{\theta}(g^{s,d}Q^{s*}+g^{r,d}Q^{r*})},\\
a_2^r=&\frac{g^{r,d}}{1+\frac{2}{\theta}(g^{s,d}P_{2}^s+g^{r,d}P_{2}^r)}-\frac{g^{r,d}}{1+\frac{2}{\theta}(g^{s,d}Q^{s*}+g^{r,d}Q^{r*})},\\
b_2^s=&P_{2}^s-Q^{s*},~~~b_2^r=P_{2}^r-Q^{r*}.
\end{align*}

\noindent If $b_1^s\ge0$ and $b_1^r\le0$, then we must have
$b_2^s\ge0$ and $b_2^r\ge0$ from $b_1^sb_2^s\ge0$ and $b_1^rb_2^r\le0$, respectively. From the format of $a_2^r$,
we further get $a_2^r\le0$, which contradicts
$a_2^rb_2^r\ge0$;
\\
If $b_1^s\le0$ and $b_1^r\ge0$, then $b_2^s\le0$ and $b_2^r\le0$.
Further, from the format of $a_2^r$, we have $a_2^r\ge0$, which contradicts
$a_2^rb_2^r\ge0$;\\
If $b_1^s\ge0$ and $b_1^r\ge0$, then from the format of $a_1^s$, we have $a_1^s\le0$,
which contradicts $a_1^sb_1^s\ge0$;\\
If $b_1^s\le0$ and
$b_1^r\le0$, then $a_1^s\ge0$, which contradicts
$a_1^sb_1^s\ge0$.\\
Thus, mini-case (1) is impossible under \textbf{Case~1}.\\
(2) $\vec{P}_{1}$ is in region $\textcircled{1}$,
$\vec{P}_{2}$ is in region $\textcircled{2}$, $\vec{Q}^*$ is
in region $\textcircled{1}$. Then,
\begin{align*}
a_1^s=&\frac{g^{s,d}}{1+\frac{2}{\theta}(g^{s,d}P_{1}^s+g^{r,d}P_{1}^r)}-\frac{g^{s,d}}{1+\frac{2}{\theta}(g^{s,d}Q^{s*}+g^{r,d}Q^{r*})},\\
a_1^r=&\frac{g^{r,d}}{1+\frac{2}{\theta}(g^{s,d}P_{1}^s+g^{r,d}P_{1}^r)}-\frac{g^{r,d}}{1+\frac{2}{\theta}(g^{s,d}Q^{s*}+g^{r,d}Q^{r*})},\\
b_1^s=&P_{1}^s-Q^{s*},~~~b_1^r=P_{1}^r-Q^{r*},\\
a_2^s=&\frac{ g^{s,r}}{1+\frac{2}{\theta} g^{s,r}P_{2}^s}-\frac{g^{s,d}}{1+\frac{2}{\theta}(g^{s,d}Q^{s*}+g^{r,d}Q^{r*})},\\
a_2^r=&-\frac{g^{r,d}}{1+\frac{2}{\theta}(g^{s,d}Q^{s*}+g^{r,d}Q^{r*})},\\
b_2^s=&P_{2}^s-Q^{s*},~~~b_2^r=P_{2}^r-Q^{r*}.
\end{align*}
\twocolfigure{100}{200}{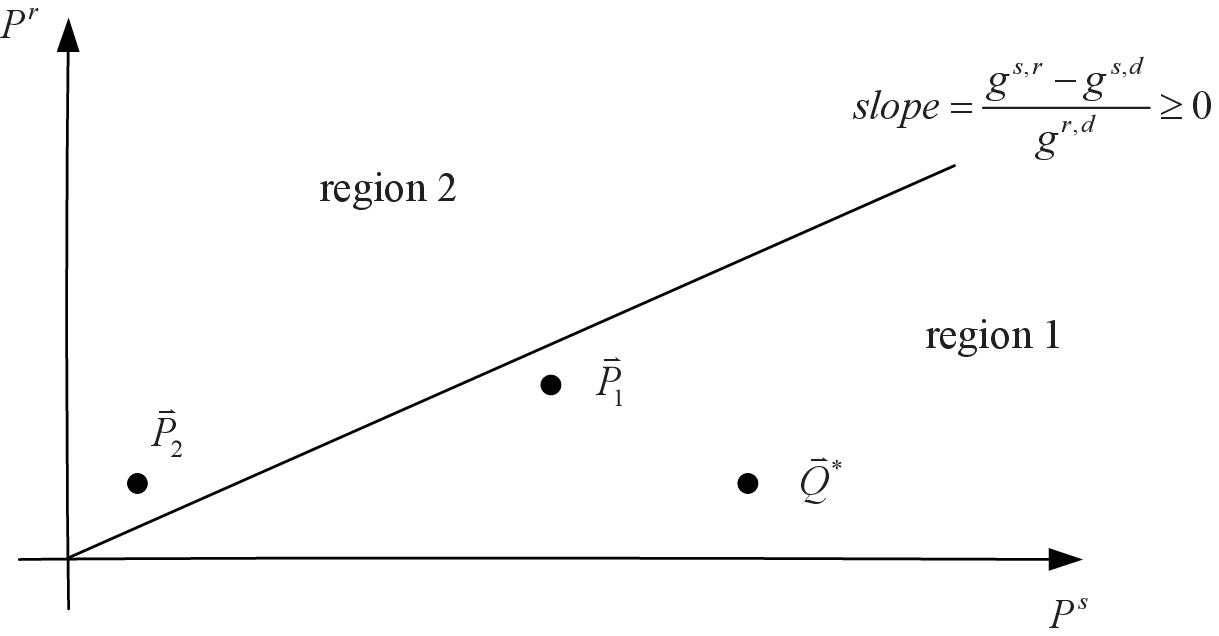}{Illustration of Power Vector Regions
for Case 1.(I).(1)}{fig:fig1}

Since $a_2^r\le0$, we get $b_2^r\le0$ from $a_2^rb_2^r\ge0$ and
$b_1^r\ge0$ from $b_1^rb_2^r\le0$. Suppose $b_1^s\ge0$, then from
the format of $a_1^s$, we have $a_1^s\le0$ which contradicts
$a_1^sb_1^s\ge0$, so $b_1^s\le0$. By $b_1^sb_2^s\ge0$, we also have
$b_2^s\le0$. With the above facts and Figure.~\ref{fig:fig1}, we
have $P_{2}^s\le P_{1}^s$, i.e., $b_2^s\le b_1^s\le0$. It is
apparent that $a_2^r\le a_1^r\le0$, and $a_1^s\le0\le a_2^s$.

Suppose $P_{1}^s=0$, then $\vec{P}_1$ is on the boundary of
region $\textcircled{1}$ and $\textcircled{2}$, and it belongs to
mini-case (4) later. Without loss of generality, let $P_{1}^s\neq0$,
and then
$L_{1}^s=0,~(a_2^s-a_1^s)(L_{2}^s-L_{1}^s)=(a_2^s-a_1^s)L_{2}^s\ge0$.
Similarly, $P_{2}^r\neq0,~L_{2}^r=0$, and
$(a_2^r-a_1^r)(L_{2}^r-L_{1}^r)=(a_1^r-a_2^r)L_{1}^r\ge0$.
Also, $a_1^sb_1^s+a_1^rb_1^r+a_2^sb_2^s+a_2^rb_2^r\le0$. Thus, we
have
\begin{align*}
&\frac{1}{4c}\Big[(\mu_{2}-\mu_{1})^2+(\nu_{2}-\nu_{1})^2\Big]-a_1^sb_2^s-a_1^rb_2^r\\
=&\frac{1}{4c}(a_2^s-a_1^s)^2+\frac{\big[c(b_1^s-b_2^s)+L_{2}^s-L_{1}^s\big]^2}{4c}+\\
&\frac{1}{2c}(a_2^s-a_1^s)(L_{2}^s-L_{1}^s)-\frac{1}{2}(a_1^sb_2^s-a_2^sb_1^s)+\\
&\frac{1}{4c}(a_2^r-a_1^r)^2+\frac{\big[c(b_1^r-b_2^r)+L_{2}^r-L_{1}^r\big]^2}{4c}+\\
&\frac{1}{2c}(a_2^r-a_1^r)(L_{2}^r-L_{1}^r)-\frac{1}{2}(a_1^rb_2^r-a_2^rb_1^r)-\\
&\frac{1}{2}(a_1^sb_1^s+a_1^rb_1^r+a_2^sb_2^s+a_2^rb_2^r)\\
\ge&\frac{1}{4c}(a_2^r-a_1^r)^2-\frac{1}{2}(a_1^sb_2^s-a_2^sb_1^s+a_1^rb_2^r-a_2^rb_1^r).
\end{align*}

We want to choose $c$ carefully such that the above term is
nonnegative. Since $\vec{Q}^*$ is feasible and $Q^{s*}$ should
be bounded by $P_{{\max}}^s$, we have
$a_1^sb_2^s-a_2^sb_1^s\le\frac{g^{s,d}Q^{s*}}{1+\frac{2}{\theta}(g^{s,d}Q^{s*}+g^{r,d}Q^{r*})}+\frac{
g^{s,r}Q^{s*}}{1+\frac{2}{\theta}
g^{s,r}P_{2}^s}\le\frac{1}{2}+
g^{s,r}P_{{\max}}^s$. Similarly,
$a_1^rb_2^r-a_2^rb_1^r\le\frac{ g^{s,r}-g^{r,d}}{g^{r,d}}$
and
$(a_2^r-a_1^r)^2=\left(\frac{g^{r,d}}{1+\frac{2}{\theta}(g^{s,d}P_{1}^s+g^{r,d}P_{1}^r)}\right)^2\ge\left(\frac{g^{r,d}}{1+\frac{2}{\theta}
g^{s,r}Q^{s*}}\right)^2$. Thus, if
$\frac{1}{4c}\ge\frac{1}{2}\frac{\frac{1}{2}+
g^{s,r}P_{{\max}}^s+\frac{
g^{s,r}-g^{r,d}}{g^{r,d}}}{\left(\frac{g^{r,d}}{1+\frac{2}{\theta}
g^{s,r}P_{{\max}}^s}\right)^2}$, i.e.,
\begin{align*}
c\le\frac{\left(\frac{g^{r,d}}{1+\frac{2}{\theta}
g^{s,r}P_{{\max}}^s}\right)^2}{1+2
g^{s,r}P_{{\max}}^s+2\frac{
g^{s,r}-g^{r,d}}{g^{r,d}}}\triangleq C_{1},
\end{align*}
we have
$a_1^sb_2^s+a_1^rb_2^r\le\frac{1}{4c}\Big[(\mu_{2}-\mu_{1})^2+(\nu_{2}-\nu_{1})^2\Big]$.

Note that this direct argument is not general to other cases since
$(a_2^r-a_1^r)^2$ may not have a positive lower bound if $
\vec{P}_{1}$ and $\vec{P}_{2}$ are very close. So, breaking
into cases is still necessary.\\
(3) $\vec{P}_{1}$ is in region $\textcircled{2}$,
$\vec{P}_{2}$ is in region $\textcircled{1}$, $\vec{Q}^*$ is
in region $\textcircled{1}$.
\begin{align*}
a_1^s=&\frac{ g^{s,r}}{1+\frac{2}{\theta} g^{s,r}P_{1}^s}-\frac{g^{s,d}}{1+\frac{2}{\theta}(g^{s,d}Q^{s*}+g^{r,d}Q^{r*})},\\
a_1^r=&-\frac{g^{r,d}}{1+\frac{2}{\theta}(g^{s,d}Q^{s*}+g^{r,d}Q^{r*})},\\
b_1^s=&P_{1}^s-Q^{s*},~~~b_1^r=P_{1}^r-Q^{r*},\\
a_2^s=&\frac{g^{s,d}}{1+\frac{2}{\theta}(g^{s,d}P_{2}^s+g^{r,d}P_{2}^r)}-\frac{g^{s,d}}{1+\frac{2}{\theta}(g^{s,d}Q^{s*}+g^{r,d}Q^{r*})},\\
a_2^r=&\frac{g^{r,d}}{1+\frac{2}{\theta}(g^{s,d}P_{2}^s+g^{r,d}P_{2}^r)}-\frac{g^{r,d}}{1+\frac{2}{\theta}(g^{s,d}Q^{s*}+g^{r,d}Q^{r*})},\\
b_2^s=&P_{2}^s-Q^{s*},~~~b_2^r=P_{2}^r-Q^{r*}.
\end{align*}
\twocolfigure{100}{200}{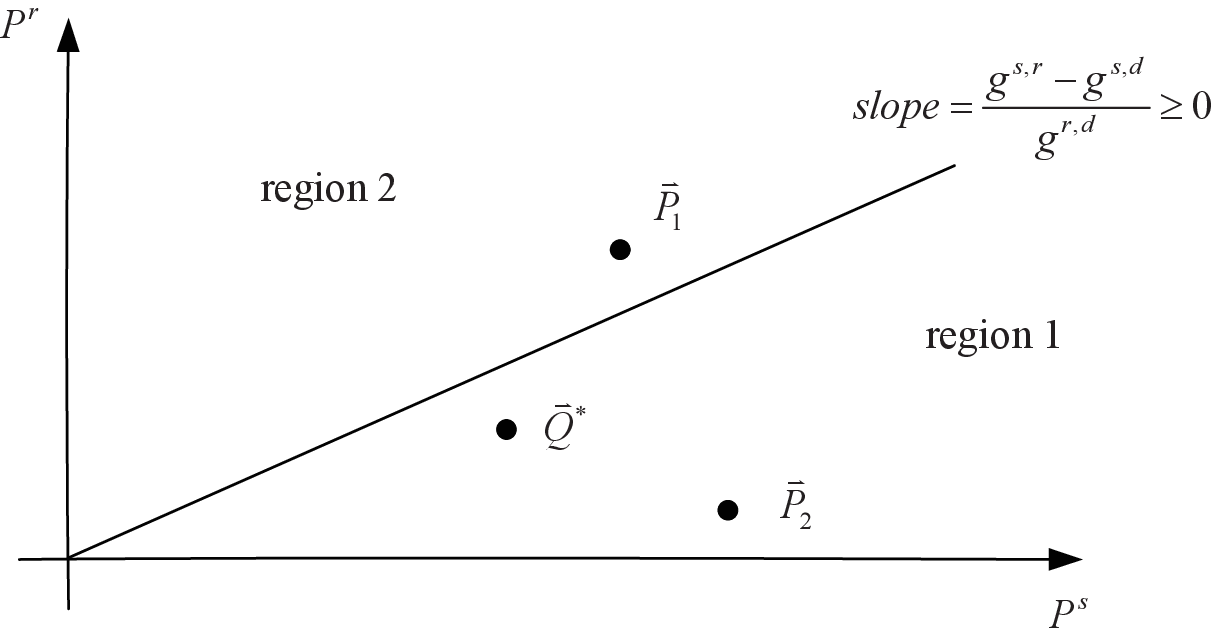}{Illustration of Power Vector Regions
for Case 1.(I).(2)}{fig:fig2}

Since $a_1^r\le0$, we have $b_1^r\ge0$ from $a_1^rb_1^r\le0$, and $a_2^r\ge 0$ from
$a_1^ra_2^r\le0$.
Suppose $b_1^s\le0$, from the format of $a_1^s$, we have $a_1^s\ge0$, which
contradicts $a_1^sb_1^s\ge0$, so $b_1^s\ge0$. Further, we have $a_1^s\ge0$, $a_2^s\le0$,
$b_2^s\ge0$ by $a_1^sb_1^s\ge0$, $a_1^sa_2^s\le0$ and $b_1^sb_2^s\ge0$, respectively. Suppose $b_2^r\ge0$, by the format
of $a_2^r$, $a_2^r\le0$, which contradicts $a_2^rb_2^r\ge0$, so
$b_2^r\le0$ and $a_2^r\le0$. If $P_{1}^s=0$, we have $Q^{s*}=0$ from $b_1^s$, which leads to triviality. Similar with mini-case (2), to avoid triviality, let
$P_{1}^r\ne0$, $P_{2}^s\ne0$, then
$L_{1}^s=L_{1}^r=L_{2}^s=0$ and
$(a_2^r-a_1^r)(L_{2}^r-L_{1}^r)=(a_2^r-a_1^r)L_{2}^r\ge0$.

In this case, in order to apply the direct argument as in mini-case
(2), we need a constant bound for
$P_{1}^s,~P_{1}^r,~P_{2}^s$. Since $a_1^s=\frac{
g^{s,r}-g^{s,d}+\frac{2}{\theta}
g^{s,r}g^{s,d}(Q^{s*}-P_{1}^s)+\frac{2}{\theta}
g^{s,r}g^{r,d}Q^{r*}}{(1+\frac{2}{\theta}
g^{s,r}P_{1}^s)(1+\frac{2}{\theta}(g^{s,d}Q^{s*}+g^{r,d}Q^{r*}))}\ge0$,
we obtain
\begin{align*}
P_{1}^s\le&\frac{\theta}{2 g^{s,r}g^{s,d}}( g^{s,r}-g^{s,d}+\frac{2}{\theta} g^{s,r}g^{s,d}Q^{s*}+\frac{2}{\theta} g^{s,r}g^{r,d}Q^{r*})\\
\le&\frac{1}{2 g^{s,r}g^{s,d}}( g^{s,r}-g^{s,d})+P_{{\max}}^s+P_{{\max}}^r\triangleq X_{1}^s
\end{align*}
\noindent which is a constant bound for $P_{1}^s$. Note that
\begin{align*}
&a_1^sb_2^s+a_1^rb_2^r\\
=&\Big[\big(a_1^s-a_2^s\big)-c\gamma_{s}^0\big(b_1^s-b_2^s\big)\Big]\big(b_2^s-b_1^s\big)+\\
&a_2^sb_2^s+a_1^sb_1^s-a_2^sb_1^s-c\gamma_{s}^0\big(b_2^s-b_1^s\big)^2+\\
&\Big[\big(a_1^r-a_2^r\big)-c\gamma_{r}^0\big(b_1^r-b_2^r\big)\Big]\big(b_2^r-b_1^r\big)+\\
&a_2^rb_2^r+a_1^rb_1^r-a_2^rb_1^r-c\gamma_{r}^0\big(b_2^r-b_1^r\big)^2\\
=&(\mu_2-\mu_1)\big(b_2^s-b_1^s\big)-c\gamma_{s}^0\big(b_2^s-b_1^s\big)^2+\\
&(\nu_2-\nu_1)\big(b_2^r-b_1^r\big)-c\gamma_{r}^0\big(b_2^r-b_1^r\big)^2+\\
&a_2^sb_2^s+a_1^sb_1^s-a_2^sb_1^s+a_2^rb_2^r+a_1^rb_1^r-a_2^rb_1^r\\
\leq &\frac{1}{4c\gamma_{r}^0}(\mu_2-\mu_1)^2+\frac{1}{4c\gamma_{s}^0}(\nu_2-\nu_1)^2+\\
&a_2^sb_2^s+a_1^sb_1^s-a_2^sb_1^s+a_2^rb_2^r+a_1^rb_1^r-a_2^rb_1^r\\
\end{align*}

If
$a_2^sb_2^s+a_1^sb_1^s-a_2^sb_1^s+a_2^rb_2^r+a_1^rb_1^r-a_2^rb_1^r\le0$,
we are done, so we assume
$a_2^sb_2^s+a_1^sb_1^s-a_2^sb_1^s+a_2^rb_2^r+a_1^rb_1^r-a_2^rb_1^r\ge0$.
Recall that $a_2^s\le a_1^s,~a_1^r\le a_2^r$ and $a_2^sb_2^s\le0$, we have
\begin{align*}
&(a_1^s-a_2^s)b_1^s+a_2^rb_2^r\ge-a_2^sb_2^s+(a_2^r-a_1^r)b_1^r\ge-a_2^sb_2^s\\
\ge&\frac{(P_{2}^s-P_{{\max}}^s)g^{s,d}}{1+\frac{2}{\theta}(g^{s,d}P_{{\max}}^s+g^{r,d}P_{{\max}}^r)}-\frac{g^{s,d}P_{2}^s}{1+\frac{2}{\theta}(g^{s,d}P_{2}^s+g^{r,d}P_{2}^r)}\\
\ge&\frac{g^{s,d}P_{2}^s}{1+\frac{2}{\theta}(g^{s,d}P_{{\max}}^s+g^{r,d}P_{{\max}}^r)}-\\
&\frac{g^{s,d}P_{{\max}}^s}{1+\frac{2}{\theta}(g^{s,d}P_{{\max}}^s+g^{r,d}P_{{\max}}^r)}-\frac{\theta}{2}
\end{align*}

Also, $(a_1^s-a_2^s)b_1^s+a_2^rb_2^r\le
g^{s,r}P_{1}^s+g^{r,d}P_{{\max}}^r=
g^{s,r}X_{1}^s+g^{r,d}P_{{\max}}^r$. Combined with
the above inequality, we have
\begin{align*}
P_{2}^s\le&\left[1+\frac{2}{\theta}(g^{s,d}P_{{\max}}^s)\right]\times\bigg[g^{s,r}X_{1}^s+g^{r,d}P_{{\max}}^r\\
&+\frac{1}{2}+\frac{g^{s,d}P_{{\max}}^s}{1+\frac{2}{\theta}(g^{s,d}P_{{\max}}^s+g^{r,d}P_{{\max}}^r)}\bigg]\times\frac{1}{g^{s,d}}\triangleq X_{2}^s
\end{align*}
\noindent which is a constant bound for $P_{2}^s$. Further,
\begin{align*}
&(a_1^s-a_2^s)b_1^s+a_2^rb_2^r\ge(a_2^r-a_1^r)b_1^r\\
\ge&\frac{g^{r,d}}{1+\frac{2}{\theta}(g^{s,d}X_{2}^s+g^{r,d}P_{{\max}}^r)}P_{1}^r,
\end{align*}
\noindent we then obtain
\begin{align*}
P_{1}^r\le&\left[1+\frac{2}{\theta}(g^{s,d}X_{2}^s+g^{r,d}P_{{\max}}^r)\right](g^{s,r}X_{1}^s+g^{r,d}P_{{\max}}^r)\frac{1}{g^{r,d}}\\
&\triangleq
X_{1}^r
\end{align*}
\noindent which is a constant bound for $P_{1}^r$. Now we can use
the direct method as in mini-case (2),
$a_1^sb_2^s-a_2^sb_1^s+a_1^rb_2^r-a_2^rb_1^r\le
g^{s,r}X_{2}^s+g^{s,d}X_{1}^s+g^{r,d}P_{{\max}}^r+g^{r,d}X_{1}^r$
and
$(a_2^r-a_1^r)^2\ge(\frac{g^{r,d}}{1+\frac{2}{\theta}(g^{s,d}X_{2}^s+g^{r,d}X_{2}^r)})^2$.
Thus, if $\frac{1}{4c}\ge\frac{
g^{s,r}X_{2}^s+g^{s,d}X_{1}^s+g^{r,d}P_{{\max}}^r+g^{r,d}X_{1}^r}{2(\frac{g^{r,d}}{1+\frac{2}{\theta}(g^{s,d}X_{2}^s+g^{r,d}X_{2}^r)})^2}$,
i.e.,
\begin{align*}
c\le&\frac{(\frac{g^{r,d}}{1+\frac{2}{\theta}(g^{s,d}X_{2}^s+g^{r,d}X_{2}^r)})^2}{2( g^{s,r}X_{2}^s+g^{s,d}X_{1}^s+g^{r,d}P_{{\max}}^r+g^{r,d}X_{1}^r)}\triangleq C_{2},
\end{align*}
\noindent we have
$a_1^sb_2^s+a_1^rb_2^r\le\frac{1}{4c}\Big[\big(\mu_{2}-\mu_{1}\big)^2+\big(\nu_{2}-\nu_{1}\big)^2\Big]$.\\
(4) $\vec{P}_{1}$ is in region $\textcircled{2}$,
$\vec{P}_{2}$ is in region $\textcircled{2}$, $\vec{Q}^*$ is
in region $\textcircled{1}$.
\begin{align*}
a_1^s=&\frac{ g^{s,r}}{1+\frac{2}{\theta} g^{s,r}P_{1}^s}-\frac{g^{s,d}}{1+\frac{2}{\theta}(g^{s,d}Q^{s*}+g^{r,d}Q^{r*})},\\
a_1^r=&-\frac{g^{r,d}}{1+\frac{2}{\theta}(g^{s,d}Q^{s*}+g^{r,d}Q^{r*})},\\
b_1^s=&P_{1}^s-Q^{s*},~~~b_1^r=P_{1}^r-Q^{r*},\\
a_2^s=&\frac{ g^{s,r}}{1+\frac{2}{\theta} g^{s,r}P_{2}^s}-\frac{g^{s,d}}{1+\frac{2}{\theta}(g^{s,d}Q^{s*}+g^{r,d}Q^{r*})},\\
a_2^r=&-\frac{g^{r,d}}{1+\frac{2}{\theta}(g^{s,d}Q^{s*}+g^{r,d}Q^{r*})},\\
b_2^s=&P_{2}^s-Q^{s*},~~~b_2^r=P_{2}^r-Q^{r*}.
\end{align*}

Since $a_1^r\le0,~a_2^r\le0$, then $b_1^r\ge0$ and $b_2^r\le0$.
Suppose $b_1^s\le0$, then $a_1^s\ge0$ which contradicts
$a_1^sb_1^s\ge0$, so $b_1^s\ge0,~b_2^s\ge0$ and $a_2^s\le0$. Now
$b_2^s\ge0$ and $b_2^r\le0$, then it is impossible to place
$\vec{P}_{2}$ in region $\textcircled{2}$ and $\vec{Q}^*$ in
region $\textcircled{1}$ at the same time. So mini-case (4) is
impossible
under \textbf{Case~1}.\\

\noindent(5) $\vec{P}_{1}$ is in region $\textcircled{1}$,
$\vec{P}_{2}$ is in region $\textcircled{1}$, $\vec{Q}^*$ is
in region $\textcircled{2}$.
\begin{align*}
a_1^s=&\frac{g^{s,d}}{1+\frac{2}{\theta}(g^{s,d}P_{1}^s+g^{r,d}P_{1}^r)}-\frac{ g^{s,r}}{1+\frac{2}{\theta} g^{s,r}Q^{s*}},\\
a_1^r=&\frac{g^{r,d}}{1+\frac{2}{\theta}(g^{s,d}P_{1}^{s}+g^{r,d}P_{1}^{r})},\\
b_1^s=&P_{1}^s-Q^{s*},~~~b_1^r=P_{1}^r-Q^{r*},\\
a_2^s=&\frac{g^{s,d}}{1+\frac{2}{\theta}(g^{s,d}P_{2}^s+g^{r,d}P_{2}^r)}-\frac{ g^{s,r}}{1+\frac{2}{\theta} g^{s,r}Q^{s*}},\\
a_2^r=&\frac{g^{r,d}}{1+\frac{2}{\theta}(g^{s,d}P_{2}^{s}+g^{r,d}P_{2}^{r})},\\
b_2^s=&P_{2}^s-Q^{s*},~~~b_2^r=P_{2}^r-Q^{r*}.
\end{align*}

Since $a_1^r\ge0$ and $a_2^r\ge0$, then $b_1^r\le0$ and $b_2^r\ge0$.
Suppose $b_1^s\ge0$, then $a_1^s\le0$ which contradicts
$a_1^sb_1^s\ge0$, so $b_1^s\le0,~b_2^s\le0$ and $a_2^s\ge0$. Now
$b_2^s\le0$ and $b_2^r\ge0$, it is impossible to place
$\vec{P}_{2}$ in region $\textcircled{1}$ and $\vec{Q}^*$ in
region $\textcircled{2}$. So, mini-case (5) is impossible under \textbf{Case~1}.\\
(6) $\vec{P}_{1}$ is in region $\textcircled{1}$,
$\vec{P}_{2}$ is in region $\textcircled{2}$, $\vec{Q}^*$ is
in region $\textcircled{2}$.\\
In this case $a_2^r=0$. By using the result $a_1^sa_2^s\le0$, we have
\begin{align*}
&a_1^sb_2^s+a_1^rb_2^r\\
\le&a_1^sb_2^s-\frac{a_1^sa_2^s}{c\gamma_{s}^0}+a_1^rb_2^r-\frac{a_1^ra_2^r}{c\gamma_{r}^0}\\
=&\frac{1}{c\gamma_{s}^0}\bigg\{\left[(a_1^s-a_2^s)-c\gamma_{s}^0(b_1^s-b_2^s)\right]a_1^s+(c\gamma_{s}^0b_1^s-a_1^s)a_1^s\bigg\}+\\
&\frac{1}{c\gamma_{r}^0}\bigg\{\left[(a_1^r-a_2^r)-c\gamma_{r}^0(b_1^r-b_2^r)\right]a_1^r+(c\gamma_{r}^0b_1^r-a_1^r)a_1^r\bigg\}\\
\le&\frac{1}{c\gamma_{s}^0}\bigg\{(\mu_{2}-\mu_{1})a_1^s-(a_1^s)^2\bigg\}\\
&+\frac{1}{c\gamma_{r}^0}\bigg\{(\nu_{2}-\nu_{1})a_1^r-(a_1^r)^2\bigg\}+(a_1^sb_1^s+a_1^rb_1^r)\\
\le&\frac{1}{4c}\Big[(\mu_{2}-\mu_{1})^2+(\nu_{2}-\nu_{1})^2\Big],
\end{align*}
where in the last step, we have used $\gamma_{s}^0=1+\gamma_{s}\geq1$,
$\gamma_{r}^0=1+\gamma_{r}\geq1$, and $a_1^sb_1^s+a_1^rb_1^r\le0$.\\
(7) $\vec{P}_{1}$ is in region $\textcircled{2}$,
$\vec{P}_{2}$ is in region $\textcircled{1}$, $\vec{Q}^*$ is
in region $\textcircled{2}$.\\
In this case $a_1^r=0$, then from
$a_1^sb_1^s+a_1^rb_1^r\le0$, we have $a_1^sb_1^s\le0$. Further, since $a_1^sb_2^s\geq0$, we have $b_1^sb_2^s\le0$. In view of the result $b_1^rb_2^r\le0$, we then have
\begin{align*}
&a_1^sb_2^s+a_1^rb_2^r\\
\le&a_1^sb_2^s-c\gamma_{s}^0b_1^sb_2^s+a_1^rb_2^r-c\gamma_{r}^0b_1^rb_2^r\\
=&\left[(a_1^s-a_2^s)-c\gamma_{s}^0(b_1^s-b_2^s)\right]b_2^s\\
&+\left[(a_1^r-a_2^r)-c\gamma_{r}^0(b_1^r-b_2^r)\right]b_2^r\\
&+(a_2^s-c\gamma_{s}^0b_2^s)b_2^s+(a_2^r-c\gamma_{r}^0b_2^r)b_2^r\\
\le&(\mu_{2}-\mu_{1})b_2^s-c\gamma_{s}^0(b_2^s)^2+\\
&(\nu_{2}-\nu_{1})b_2^r-c\gamma_{r}^0(b_2^r)^2+(a_2^sb_2^s+a_2^rb_2^r)\\
\le&\frac{1}{4c}\Big[(\mu_{2}-\mu_{1})^2+(\nu_{2}-\nu_{1})^2\Big].
\end{align*}
(8) $\vec{P}_{1}$ is in region $\textcircled{2}$,
$\vec{P}_{2}$ is in region $\textcircled{2}$, $\vec{Q}^*$ is
in region $\textcircled{2}$.\\
$a_1^rb_1^r=0,~a_2^rb_2^r=0,~a_1^sb_1^s\le0,~a_2^sb_2^s\le0$, then
$a_1^sb_2^s+a_1^rb_2^r=a_1^sb_2^s$, and the techniques
in~\cite{Lin06} applies.\\
(II) If $a_1^sb_2^s\le0$ and $a_1^rb_2^r\le0$, then is is trivial.\\
(III) If $a_1^sb_2^s>0$ and $a_1^rb_2^r<0$, then let
$\gamma^s\triangleq-\frac{a_2^s}{c\gamma_{s}^0b_2^s}\ge0$
and
$\gamma^r\triangleq-\frac{a_2^r}{c\gamma_{r}^0b_2^r}\le0$.
\begin{align*}
&a_1^sb_2^s+a_1^rb_2^r\le(1+\gamma^s)a_1^sb_2^s+(1+\gamma^r)a_1^rb_2^r\\
=&\frac{1}{c\gamma_{s}^0}\Big\{\big[(a_1^s-a_2^s)-c\gamma_{s}^0(b_1^s-b_2^s)\big]a_1^s+(c\gamma_{s}^0b_1^s-a_1^s)a_1^s\Big\}+\\
&\frac{1}{c\gamma_{r}^0}\Big\{\big[(a_1^r-a_2^r)-c\gamma_{r}^0(b_1^r-b_2^r)\big]a_1^r+(c\gamma_{r}^0b_1^r-a_1^r)a_1^r\Big\}\\
\le&\frac{1}{4c}\big[(\mu_{2}-\mu_{1})^2+(\nu_{2}-\nu_{1})^2\big].
\end{align*}
(IV) If $a_1^sb_2^s<0$ and $a_1^rb_2^r>0$, it can be dealt with
similarly as above.\\
\textbf{Case~2}: When $a_1^sb_1^s\le0,~a_1^rb_1^r\ge0,~a_2^sb_2^s\ge0,~a_2^rb_2^r\le0$.\\
(I) If $a_1^sb_2^s\ge0$ and $a_1^rb_2^r\ge0$, then we have
$a_1^sa_2^s\ge0,~b_1^sb_2^s\le0,~a_1^ra_2^r\le0,~b_1^rb_2^r\ge0$.
Now we further break into mini-cases:\\
(1) $\vec{P}_{1}$ is in region $\textcircled{1}$,
$\vec{P}_{2}$ is in region $\textcircled{1}$, $\vec{Q}^*$ is
in region $\textcircled{1}$.
\begin{align*}
a_1^s=&\frac{g^{s,d}}{1+\frac{2}{\theta}(g^{s,d}P_{1}^s+g^{r,d}P_{1}^r)}-\frac{g^{s,d}}{1+\frac{2}{\theta}(g^{s,d}Q^{s*}+g^{r,d}Q^{r*})},\\
a_1^r=&\frac{g^{r,d}}{1+\frac{2}{\theta}(g^{s,d}P_{1}^s+g^{r,d}P_{1}^r)}-\frac{g^{r,d}}{1+\frac{2}{\theta}(g^{s,d}Q^{s*}+g^{r,d}Q^{r*})},\\
b_1^s=&P_{1}^s-Q^{s*},~~~b_1^r=P_{1}^r-Q^{r*}\\
a_2^s=&\frac{g^{s,d}}{1+\frac{2}{\theta}(g^{s,d}P_{2}^s+g^{r,d}P_{2}^r)}-\frac{g^{s,d}}{1+\frac{2}{\theta}(g^{s,d}Q^{s*}+g^{r,d}Q^{r*})},\\
a_2^r=&\frac{g^{r,d}}{1+\frac{2}{\theta}(g^{s,d}P_{2}^s+g^{r,d}P_{2}^r)}-\frac{g^{r,d}}{1+\frac{2}{\theta}(g^{s,d}Q^{s*}+g^{r,d}Q^{r*})},\\
b_2^s=&P_{2}^s-Q^{s*},~~~b_2^r=P_{2}^r-Q^{r*}.
\end{align*}

\noindent If $b_1^s\ge0$ and $b_1^r\le0$, then $b_2^s\le0$ and
$b_2^r\le0$, and we further have $a_2^s\ge0$ and $a_2^r\ge0$, which
contradicts
$a_2^sb_2^s\ge0$;\\
If $b_1^s\le0$ and $b_1^r\ge0$, then $b_2^s\ge0$ and $b_2^r\ge0$,
and we further have $a_2^s\le0$ and $a_2^r\le0$, which contradicts
$a_2^sb_2^s\ge0$;\\
If $b_1^s\ge0$ and $b_1^r\ge0$, then $a_1^s\le0$ and $a_1^r\le0$,
which contradicts $a_1^rb_1^r\ge0$;\\
If $b_1^s\le0$ and $b_1^r\le0$, then $a_1^s\ge0$ and $a_1^r\ge0$,
which contradicts $a_1^rb_1^r\ge0$.\\
So mini-case (1) is impossible under \textbf{Case~2}.\\
(2) $\vec{P}_{1}$ is in region $\textcircled{1}$,
$\vec{P}_{2}$ is in region $\textcircled{2}$, $\vec{Q}^*$ is
in region $\textcircled{1}$.
\begin{align*}
a_1^s=&\frac{g^{s,d}}{1+\frac{2}{\theta}(g^{s,d}P_{1}^s+g^{r,d}P_{1}^r)}-\frac{g^{s,d}}{1+\frac{2}{\theta}(g^{s,d}Q^{s*}+g^{r,d}Q^{r*})},\\
a_1^r=&\frac{g^{r,d}}{1+\frac{2}{\theta}(g^{s,d}P_{1}^s+g^{r,d}P_{1}^r)}-\frac{g^{r,d}}{1+\frac{2}{\theta}(g^{s,d}Q^{s*}+g^{r,d}Q^{r*})},\\
b_1^s=&P_{1}^s-Q^{s*},~~~b_1^r=P_{1}^r-Q^{r*}\\
a_2^s=&\frac{ g^{s,r}}{1+\frac{2}{\theta} g^{s,r}P_{2}^s}-\frac{g^{s,d}}{1+\frac{2}{\theta}(g^{s,d}Q^{s*}+g^{r,d}Q^{r*})},\\
a_2^r=&-\frac{g^{r,d}}{1+\frac{2}{\theta}(g^{s,d}Q^{s*}+g^{r,d}Q^{r*})},\\
b_2^s=&P_{2}^s-Q^{s*},~~~b_2^r=P_{2}^r-Q^{r*}.
\end{align*}

\twocolfigure{100}{200}{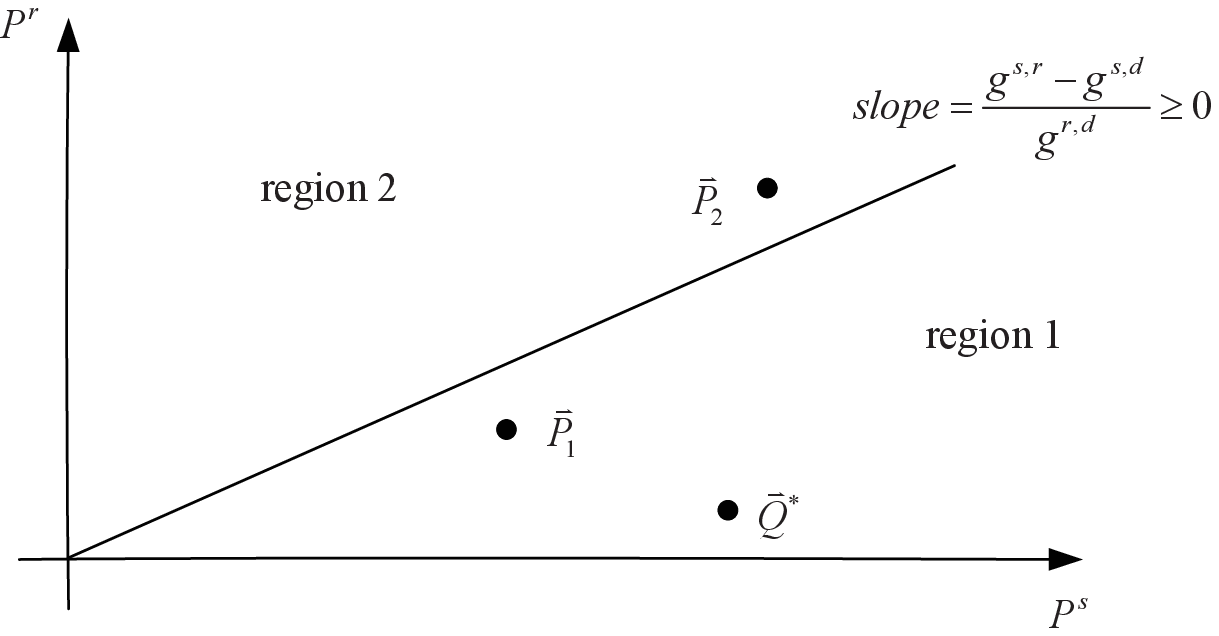}{Illustration of Power Vector Regions
for Case 2.(I).(2)}{fig:fig3}

Since $a_2^r\le0$, we get $b_2^r\ge0$ by $a_2^rb_2^r\le0$ and
$b_1^r\ge0$ by $b_1^rb_2^r\ge0$. Suppose $b_1^s\ge0$, then
$a_1^s\le0$ and $a_1^r\le0$, which contradicts $a_1^rb_1^r\ge0$, so
$b_1^s\le0$. Then, $b_2^s\ge0,~a_1^s\ge0,~a_1^r\ge0,~a_2^s\ge0$.
Without loss of generality, let
$P_{1}^s>0,~P_{1}^r>0,~P_{2}^s>0,~P_{2}^r>0$, so
$L_{1}^s=L_{2}^s=L_{1}^r=L_{2}^r=0$. Note that
$P_{1}^s\le Q^{s*}\le P_{{\max}}^s$ and
$P_{1}^r\le\frac{
g^{s,r}-g^{s,d}}{g^{r,d}}P_{1}^s=\frac{
g^{s,r}-g^{s,d}}{g^{r,d}}P_{{\max}}^s$. Since
$a_2^s=\frac{ g^{s,r}-g^{s,d}+\frac{2}{\theta}
g^{s,r}g^{s,d}(Q^{s*}-P_{2}^s)+\frac{2}{\theta}
g^{s,r}g^{r,d}Q^{r*}}{(1+\frac{2}{\theta}
g^{s,r}P_{2}^s)(1+\frac{2}{\theta}(g^{s,d}Q^{s*}+g^{r,d}Q^{r*}))}\ge0$,
we obtain
\begin{align*}
P_{2}^s\le&\frac{\theta}{2 g^{s,r}g^{s,d}}\Big( g^{s,r}-g^{s,d}+\frac{2}{\theta} g^{s,r}g^{s,d}Q^{s*}+\frac{2}{\theta} g^{s,r}g^{r,d}Q^{r*}\Big)\\
\le&\frac{1}{2 g^{s,r}g^{s,d}}(g^{s,r}-g^{s,d})+P_{{\max}}^s+P_{{\max}}^r\triangleq X_{1}^s
\end{align*}
\noindent which is a constant bound for $P_{2}^s$. Use the
similar idea as in \textbf{Case~1}. Without loss of generality,
assume
$a_2^sb_2^s+a_1^sb_1^s-a_2^sb_1^s+a_2^rb_2^r+a_1^rb_1^r-a_2^rb_1^r\ge0$,
then
$a_2^sb_2^s-a_2^sb_1^s+a_1^rb_1^r-a_2^rb_1^r\ge-a_1^sb_1^s-a_2^rb_2^r\ge-a_2^rb_2^r\ge\frac{g^{r,d}(P_{2}^r-P_{{\max}}^r)}{1+\frac{2}{\theta}(g^{s,d}P_{{\max}}^s+g^{r,d}P_{{\max}}^r)}$.
Also,
$a_2^s(b_2^s-b_1^s)+(a_1^r-a_2^r)b_1^r\le\frac{\theta}{2}+g^{r,d}\frac{
g^{s,r}-g^{s,d}}{g^{r,d}}P_{{\max}}^s=\frac{\theta}{2}+(
g^{s,r}-g^{s,d})P_{{\max}}^s$, then
\begin{align*}
P_{2}^r\le&\frac{1}{g^{r,d}}\left[\frac{1}{2}+( g^{s,r}-g^{s,d})P_{{\max}}^s\right]\times\\
&\left[1+\frac{2}{\theta}(g^{s,d}P_{{\max}}^s+g^{r,d}P_{{\max}}^r)\right]+P_{{\max}}^r\\
\triangleq &X_{2}^r
\end{align*}
\noindent which is a constant bound for $P_{2}^r$. Thus, if
$\frac{1}{4c}\ge\frac{1}{2}\frac{g^{s,d}X_{2}^s+
g^{s,r}P_{{\max}}^s+g^{r,d}X_{2}^r+(
g^{s,r}-g^{s,d})P_{{\max}}^s}{(\frac{g^{r,d}}{1+\frac{2}{\theta}
g^{s,r}P_{{\max}}^s})^2}\triangleq C_{3}$.\\
(3) $\vec{P}_{1}$ is in region $\textcircled{2}$,
$\vec{P}_{2}$ is in region $\textcircled{1}$, $\vec{Q}^*$ is
in region $\textcircled{1}$.
\begin{align*}
a_1^s=&\frac{ g^{s,r}}{1+\frac{2}{\theta} g^{s,r}P_{1}^s}-\frac{g^{s,d}}{1+\frac{2}{\theta}(g^{s,d}Q^{s*}+g^{r,d}Q^{r*})},\\
a_1^r=&-\frac{g^{r,d}}{1+\frac{2}{\theta}(g^{s,d}Q^{s*}+g^{r,d}Q^{r*})},\\
b_1^s=&P_{1}^s-Q^{s*},~~~b_1^r=P_{1}^r-Q^{r*}\\
a_2^s=&\frac{g^{s,d}}{1+\frac{2}{\theta}(g^{s,d}P_{2}^s+g^{r,d}P_{2}^r)}-\frac{g^{s,d}}{1+\frac{2}{\theta}(g^{s,d}Q^{s*}+g^{r,d}Q^{r*})},\\
a_2^r=&\frac{g^{r,d}}{1+\frac{2}{\theta}(g^{s,d}P_{2}^s+g^{r,d}P_{2}^r)}-\frac{g^{r,d}}{1+\frac{2}{\theta}(g^{s,d}Q^{s*}+g^{r,d}Q^{r*})},\\
b_2^s=&P_{2}^s-Q^{s*},~~~b_2^r=P_{2}^r-Q^{r*}.
\end{align*}

\twocolfigure{100}{200}{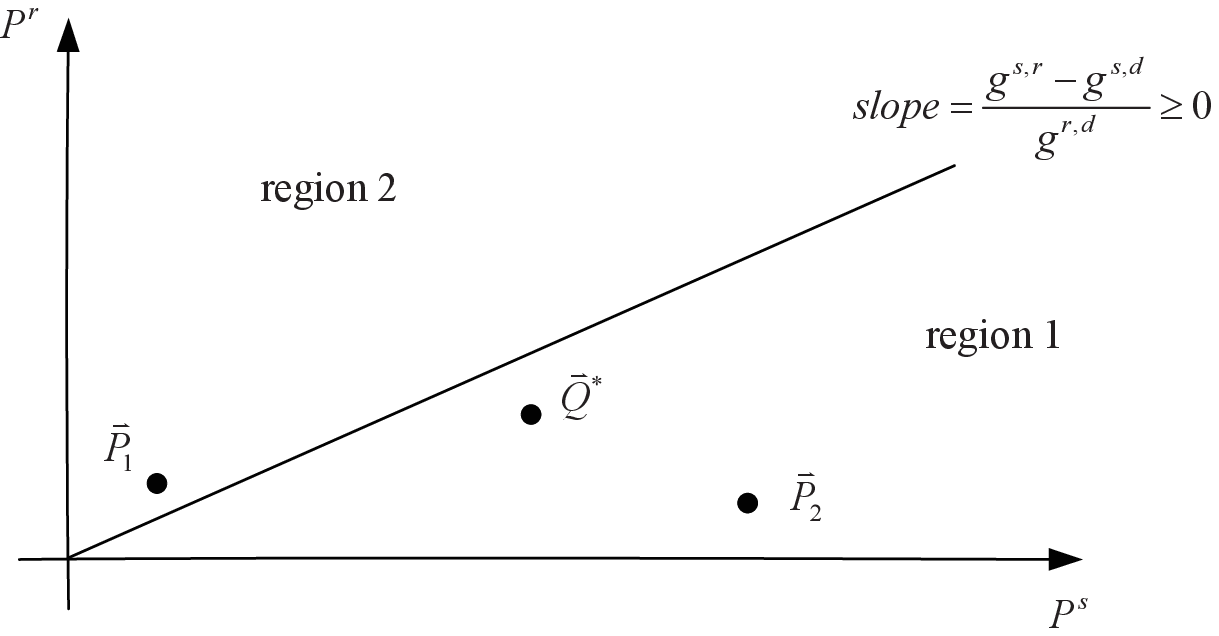}{Illustration of Power Vector Regions
for Case 2.(I).(3)}{fig:fig4}

Since $a_1^r\le0$, we get $b_1^r\le0$ and $b_2^r\le0$. Suppose
$b_2^s\le0$, then $a_2^s\ge0$ and $a_2^r\ge0$, which contradicts
$a_2^sb_2^s\ge0$, so
$b_2^s\ge0,~b_1^s\le0,~a_1^s\ge0,~a_2^s\ge0,~a_2^r\ge0$. Without
loss of generality, let $P_{1}^r>0$ and $P_{2}^s>0$, so
$L_{1}^r=0$ and $L_{2}^s=0$. Further, $P_{1}^s\le
P_{2}^s$, then $a_2^s\le a_1^s$, so
$(a_2^s-a_1^s)(-L_{1}^s)\ge0$ and $(a_2^r-a_1^r)L_{2}^r\ge0$.
Note that $P_{1}^s\le Q^{s*}\le P_{{\max}}^s$,
$P_{1}^r\le Q^{r*}\le P_{{\max}}^r$, $P_{2}^r\le
Q^{r*}\le P_{{\max}}^r$, and
$a_2^s=g^{s,d}\frac{\frac{2}{\theta}[g^{s,d}(Q^{s*}-P_{2}^s)+g^{r,d}(Q^{r*}-P_{2}^r)]}{[1+\frac{2}{\theta}(g^{s,d}P_{2}^s+g^{r,d}P_{2}^r)][1+\frac{2}{\theta}(g^{s,d}Q^{s*}+g^{r,d}Q^{r*})]}\ge0$,
then $P_{2}^s\le
Q^{s*}+\frac{g^{r,d}}{g^{s,d}}Q^{r*}=P_{{\max}}^s+\frac{g^{r,d}}{g^{s,d}}P_{{\max}}^r$.
Thus, if $\frac{1}{4c}\ge\frac{1}{2}\frac{
g^{s,r}(P_{{\max}}^s+\frac{g^{r,d}}{g^{s,d}}P_{{\max}}^r)+g^{s,d}P_{{\max}}^s+g^{r,d}P_{{\max}}^r+g^{r,d}P_{{\max}}^r}{\left(\frac{g^{r,d}}{1+\frac{2}{\theta}\left[g^{s,d}(P_{{\max}}^s+\frac{g^{r,d}}{g^{s,d}}P_{{\max}}^r)+g^{r,d}P_{{\max}}^r\right]}\right)^2}$,
i.e.,
\begin{align*}
c\le&\frac{1}{2}\frac{\left(\frac{g^{r,d}}{1+\frac{2}{\theta}(g^{s,d}P_{{\max}}^s+2g^{r,d}P_{{\max}}^r)}\right)^2}{( g^{s,r}+g^{s,d})P_{{\max}}^s+(\frac{ g^{s,r}g^{r,d}}{g^{s,d}}+2g^{r,d})P_{{\max}}^r}\nonumber\\
\triangleq&C_{4}
\end{align*}
(4) $\vec{P}_{1}$ is in region $\textcircled{2}$,
$\vec{P}_{2}$ is in region $\textcircled{2}$, $\vec{Q}^*$ is
in region $\textcircled{1}$.
\begin{align*}
a_1^s=&\frac{ g^{s,r}}{1+\frac{2}{\theta} g^{s,r}P_{1}^s}-\frac{g^{s,d}}{1+\frac{2}{\theta}(g^{s,d}Q^{s*}+g^{r,d}Q^{r*})},\\
a_1^r=&-\frac{g^{r,d}}{1+\frac{2}{\theta}(g^{s,d}Q^{s*}+g^{r,d}Q^{r*})},\\
b_1^s=&P_{1}^s-Q^{s*},~~~b_1^r=P_{1}^r-Q^{r*}\\
a_2^s=&\frac{ g^{s,r}}{1+\frac{2}{\theta} g^{s,r}P_{2}^s}-\frac{g^{s,d}}{1+\frac{2}{\theta}(g^{s,d}Q^{s*}+g^{r,d}Q^{r*})},\\
a_2^r=&-\frac{g^{r,d}}{1+\frac{2}{\theta}(g^{s,d}Q^{s*}+g^{r,d}Q^{r*})},\\
b_2^s=&P_{2}^s-Q^{s*},~~~b_2^r=P_{2}^r-Q^{r*}.
\end{align*}
$a_1^r=a_2^r<0$, this contradicts $a_1^ra_2^r\le0$. So mini-case
(4) is
impossible under \textbf{Case~2}.\\
(5) $\vec{P}_{1}$ is in region $\textcircled{1}$,
$\vec{P}_{2}$ is in region $\textcircled{1}$, $\vec{Q}^*$ is
in region $\textcircled{2}$.
\begin{align*}
a_1^s=&\frac{g^{s,d}}{1+\frac{2}{\theta}(g^{s,d}P_{1}^s+g^{r,d}P_{1}^r)}-\frac{ g^{s,r}}{1+\frac{2}{\theta} g^{s,r}Q^{s*}},\\
a_1^r=&\frac{g^{r,d}}{1+\frac{2}{\theta}(g^{s,d}P_{1}^{s}+g^{r,d}P_{1}^{r})},\\
b_1^s=&P_{1}^s-Q^{s*},~~~b_1^r=P_{1}^r-Q^{r*}\\
a_2^s=&\frac{g^{s,d}}{1+\frac{2}{\theta}(g^{s,d}P_{2}^s+g^{r,d}P_{2}^r)}-\frac{ g^{s,r}}{1+\frac{2}{\theta} g^{s,r}Q^{s*}},\\
a_2^r=&\frac{g^{r,d}}{1+\frac{2}{\theta}(g^{s,d}P_{2}^{s}+g^{r,d}P_{2}^{r})},\\
b_2^s=&P_{2}^s-Q^{s*},~~~b_2^r=P_{2}^r-Q^{r*}.
\end{align*}
$a_1^r\ge0,~a_2^r\ge0$, this contradicts $a_1^ra_2^r\le0$. So
mini-case
(5) is impossible under \textbf{Case~2}.\\
(6),(7),(8) are all similar as in \textbf{Case~1}.\\
(II) If $a_1^sb_2^s\le0$ and $a_1^rb_2^r\le0$, it is trivial;\\
(III) If $a_1^sb_2^s\ge0$ and $a_1^rb_2^r\le0$, then let
$\gamma^s\triangleq-\frac{c\gamma_{s}^0b_1^s}{a_1^s}\ge0$
and
$\gamma^r\triangleq-\frac{c\gamma_{r}^0b_1^r}{a_1^r}\le0$ and similar argument as in \textbf{Case~1} follows;\\
(IV) If $a_1^sb_2^s\le0$ and $a_1^rb_2^r\ge0$, it is similar as above.\\
\textbf{Case~3}: When
$a_1^sb_1^s\le0,~a_1^rb_1^r\ge0,~a_2^sb_2^s\le0,~a_2^rb_2^r\ge0$.\\
We only prove for the subcase when $a_1^sb_2^s>0$ and
$a_1^rb_2^r>0$. The same proof ideas as in \textbf{Case~1} can be
applied in all other subcases.

Let $\gamma^s\triangleq\frac{a_2^sb_1^s}{a_1^sb_2^s}\ge0$ and
$\gamma^r\triangleq\frac{a_2^rb_1^r}{a_1^rb_2^r}\ge0$, then
\begin{align*}
&a_1^sb_2^s+a_1^rb_2^r\\
\le&(1+\gamma^s)a_1^sb_2^s+(1+\gamma^r)a_1^rb_2^r\\
=&\left[(a_1^s-a_2^s)-c\gamma_{s}^0(b_1^s-b_2^s)\right](b_2^s-b_1^s)\\
&+\left[(a_1^r-a_2^r)-c\gamma_{r}^0(b_1^r-b_2^r)\right](b_2^r-b_1^r)\\
&+(a_1^sb_1^s+a_2^sb_2^s)+(a_1^rb_1^r+a_2^rb_2^r)\\
&-c\gamma_{s}^0(b_2^s-b_1^s)^2-c\gamma_{r}^0(b_2^r-b_1^r)^2\\
\le&(\mu_{2}-\mu_{1})(b_2^s-b_1^s)-c\gamma_{s}^0(b_2^s-b_1^s)^2\\
&+(\nu_{2}-\nu_{1})(b_2^r-b_1^r)-c\gamma_{r}^0(b_2^r-b_1^r)^2\\
&+(a_1^sb_1^s+a_1^rb_1^r)+(a_2^sb_2^s+a_2^rb_2^r)\\
\le&\frac{1}{4c}\Big[(\mu_{2}-\mu_{1})^2+(\nu_{2}-\nu_{1})^2\Big].
\end{align*}
\textbf{Case~4}: When $a_1^sb_1^s\ge0,~a_1^rb_1^r\le0,~a_2^sb_2^s\ge0,~a_2^rb_2^r\le0$.\\
The similar argument as in \textbf{Case~3} can be applied here.\\
\textbf{Case~5}: When $a_1^sb_1^s\le0,~a_1^rb_1^r\le0$.\\
We only prove for the subcase when $a_1^sb_2^s>0$ and
$a_1^rb_2^r>0$. The same proof ideas as in \textbf{Case~1} can be
applied in all other subcases. Let
$\gamma^s\triangleq-\frac{c\gamma_{s}^0b_1^s}{a_1^s}\ge0,~\gamma^r\triangleq-\frac{c\gamma_{r}^0b_1^r}{a_1^r}\ge0$,
then
\begin{align*}
&a_1^sb_2^s+a_1^rb_2^r\\
\le&(1+\gamma^s)a_1^sb_2^s+(1+\gamma^r)a_1^rb_2^r\\
=&\left[(a_1^s-a_2^s)-c\gamma_{s}^0(b_1^s-b_2^s)\right]b_2^s\\
&+\left[(a_1^r-a_2^r)-c\gamma_{r}^0(b_1^r-b_2^r)\right]b_2^r\\
&+(a_2^s-c\gamma_{s}^0b_2^s)b_2^s+(a_2^r-c\gamma_{r}^0b_2^r)b_2^r\\
\le&(\mu_{2}-\mu_{1})b_2^s-c\gamma_{s}^0(b_2^s)^2+(\nu_{2}-\nu_{1})b_2^r-c\gamma_{r}^0(b_2^r)^2+\\
&(a_2^sb_2^s+a_2^rb_2^r)\\
\le&\frac{1}{4c}\Big[(\mu_{2}-\mu_{1})^2+(\nu_{2}-\nu_{1})^2\Big].
\end{align*}
\textbf{Case~6}: $a_2^sb_2^s\le0,~a_2^rb_2^r\le0$.\\
We only prove for the subcase when $a_1^sb_2^s>0$ and
$a_1^rb_2^r>0$. The same proof ideas as in \textbf{Case~1} can be
applied in all other subcases. Let
$\gamma^s\triangleq-\frac{a_2^s}{c\gamma_{s}^0b_2^s}\ge0,~\gamma^r\triangleq-\frac{a_2^r}{c\gamma_{r}^0b_2^r}\ge0$,
then
\begin{align*}
&a_1^sb_2^s+a_1^rb_2^r\\
\le&(1+\gamma^s)a_1^sb_2^s+(1+\gamma^r)a_1^rb_2^r\\
=&\frac{1}{c\gamma_{s}^0}\bigg\{\left[(a_1^s-a_2^s)-c\gamma_{s}^0(b_1^s-b_2^s)\right]a_1^s+(c\gamma_{s}^0b_1^s-a_1^s)a_1^s\bigg\}+\\
&\frac{1}{c\gamma_{r}^0}\bigg\{\left[(a_1^r-a_2^r)-c\gamma_{r}^0(b_1^r-b_2^r)\right]a_1^r+(c\gamma_{r}^0b_1^r-a_1^r)a_1^r\bigg\}\\
\le&\frac{1}{c\gamma_{s}^0}\bigg\{(\mu_{2}-\mu_{1})a_1^s-(a_1^s)^2\bigg\}+\\
&\frac{1}{c\gamma_{r}^0}\bigg\{(\nu_{2}-\nu_{1})a_1^r-(a_1^r)^2\bigg\}+(a_1^sb_1^s+a_1^rb_1^r)\\
\le&\frac{1}{4c}\Big[(\mu_{2}-\mu_{1})^2+(\nu_{2}-\nu_{1})^2\Big].
\end{align*}
\textbf{Case~7}: When $a_1^sb_1^s>0,~a_1^rb_1^r\ge0$; or
$a_2^sb_2^s>0,~a_2^rb_2^r\ge0$. This case cannot happen, otherwise,
it contradicts (\ref{eq61}).~\qed

\else
\fi
\end{document}